%% file: manuscript.tex
\documentclass[10pt,twocolumn,twoside]{article}

\usepackage[utf8]{inputenc}
\usepackage[letterpaper,left=0.6in,top=0.5in,right=0.6in,bottom=1.0in]{geometry}
\usepackage{hyperref,amsbsy,amsmath,amsthm,amsfonts,mathrsfs,float}
\usepackage{graphicx,subfigure,tabularx,bm,soul}
\usepackage{enumitem,amssymb,lipsum,abstract,parskip}
\usepackage[labelfont=bf]{caption}
\usepackage[dvipsnames]{xcolor}
\usepackage{sectsty}
\usepackage[switch]{lineno}
\usepackage{authblk}
\usepackage{physics}
\usepackage{wrapfig}
\usepackage{pdflscape,rotating}
\newlist{todolist}{itemize}{2}
\setlist[todolist]{label=$\square$}
\usepackage{pifont}

\usepackage{bbm,mathtools,cases,multirow}
\usepackage[normalem]{ulem}

\renewcommand \thesection{\Roman{section}.}
\renewcommand \thesubsection{\Alph{subsection}.}
\allsectionsfont{\nohang\centering\normalsize\uppercase}
\subsectionfont{\nohang\centering\small}

\newcommand{\lb}{\left(}    
\newcommand{\rb}{\right)}
\newcommand{\lcb}{\left\{\, }
\newcommand{\rcb}{\,\right\} }
\newcommand{\lsb}{\left[\, }
\newcommand{\rsb}{\,\right] }

\newcommand{\Dphi}{\Delta\varphi}
\newcommand{\rhat}{\hat{\mathbf{r}}}
\newcommand{\thetahat}{\hat{\bm{\theta}}}
\newcommand{\phat}{\hat{\mathbf{p}}}
\newcommand{\ehat}{\hat{\mathbf{e}}}

\DeclareMathOperator{\arctantwo}{arctan2}
\DeclareMathOperator{\sgn}{sgn}

\graphicspath{{Figures/}}

\usepackage{ulem}
\usepackage{lipsum}

\newtheorem{theorem}{Theorem}[section]
\newtheorem{lemma}[theorem]{Lemma}
\newtheorem{proposition}[theorem]{Proposition}
\newtheorem{corollary}[theorem]{Corollary}

\theoremstyle{definition}

\renewcommand{\thetheorem}{\Roman{section}.\arabic{theorem}}

\begin{document}

\title{Directional bias of a single polarized cell under confinement}

\author[1,*]{Andreas Buttensch\"{o}n}
\author[2,3,*]{Calina Copos}

\affil[1]{\small{Department of Mathematics and Statistics, University of Massachusetts Amherst, Amherst, MA, USA}}
\affil[2]{\small{Department of Mathematics, Northeastern University, Boston, MA, USA}}
\affil[3]{\small{Department of Biology, Northeastern University, Boston, MA, USA}}
\affil[*]{\small{Co-corresponding authors: andreas.buttenschoen@umass.edu, c.copos@northeastern.edu}}


\date{}

\twocolumn[
  \maketitle
\input{abstract}
]

\input{main_text}
\textbf{Competing interests.} The authors declare no conflicts of interest.

\textbf{Data availability.} The article is a theoretical study and does not include any empirical data. The Matlab code used to simulate the model is deposited on Zenodo:~\url{https://zenodo.org/records/20753115}.

\textbf{Use of AI.} The authors used AI tools for language editing purposes, and an AI coding assistant (Anthropic's Claude) to help write and refine the Python code used to generate the figures. All final content, arguments, and conclusions remain the sole responsibility of the authors.

\textbf{Author contributions.} Both authors contributed equally in all aspects of this work including conceptualization, formal analysis, methodology, investigation, writing - original draft, writing - review and editing.

\textbf{Acknowledgments.} We thank Dr. La\"etitia Kurzawa for fruitful discussions on the biological implications of our theoretical results and E. Im for critically reading our manuscript. This work was supported in part by NSF DMS2209494 (C.C.).

\bibliographystyle{unsrt}
\bibliography{rotationalsinglet.bib}

\clearpage 
\onecolumn
\thispagestyle{empty}

\noindent\textbf{SUPPLEMENTAL TABLES:}
\begin{itemize}[itemsep=0mm]
\item[]\textbf{SUPPLEMENTAL TABLE 1}:
Table of parameter descriptions along with the values used in computer simulations and analysis.

\item[]\textbf{SUPPLEMENTAL TABLE 2}: Summary of the dynamical systems analysis results for the singlet. Each mechanism of the fundamental reduced system Eq.~\eqref{eq:gen_master} is switched on alone; $\mathbb{Z}_2$ denotes the reflection symmetry $\Dphi \to -\Dphi$ of Appendix~B, and $\ell = V_0\tau_W$.
\end{itemize}

\noindent\textbf{SUPPLEMENTAL FIGURE LEGENDS:}

\begin{itemize}[itemsep=0mm]
\item[]\textbf{SUPPLEMENTAL FIGURE 1}: (A) Sample trajectories in the absence of confinement forces. (B) Heatmaps of the outcomes for the percentage of coherent rotations and CW rotations with parameter variations in the default unbiased singlet model.

\item[]\textbf{SUPPLEMENTAL FIGURE 2}: Sample trajectories for simulated unbiased ($\mu=0$) singlets.

\item[]\textbf{SUPPLEMENTAL FIGURE 3}: Same as Fig.~S2 but for biased ($\mu=-0.5$) singlets.

\item[]\textbf{SUPPLEMENTAL FIGURE 4}: Same as Fig.~S2 (unbiased singlets), but with an easy frictional direction at an offset angle $\delta=\pi/4$ from the polarity axis. 

\item[]\textbf{SUPPLEMENTAL FIGURE 5}: Phase portraits of the singlet model with chiral cell–wall interactions modulated by parameter $\chi$. Depending on the sign of $\chi$, one rotational equilibrium loses stability. White circles denote stable (or neutrally stable) equilibria, white crosses unstable equilibria, the background colormap indicates the time-averaged winding number $n(T)$, and solid colored curves show representative trajectories.

\item[]\textbf{SUPPLEMENTAL FIGURE 6}: Two-parameter phase diagrams of the anchored singlet combined with each symmetry-breaking coupling: intrinsic bias $\mu$ (top row), chiral wall $\chi$ (middle row), and anisotropic mobility $(\alpha)$ (bottom row). Parameters are $k_W = 1.0, k_0 = 0.5$ so that $\kappa = 1/2$ and $R_{\rm disc} = 1$. White circles denote stable (or neutrally stable) equilibria, the background colormap indicates the time-averaged winding number $n(T)$, and solid colored curves show representative trajectories.

\item[]\textbf{SUPPLEMENTAL FIGURE 7}: (A) Motility outcomes shown as percentages for varying anisotropic frictional coefficient ($\alpha$) for unbiased cells migrating on pattern (v). Grey stars mark percentage of rotating cells, while blue dots indicate clockwise percentage. (B-C) Sample trajectories for simulated unbiased singlets on (B) pattern (v) and (C) pattern (vi).

\end{itemize}

\noindent{\textbf{SUPPORTING MOVIES}}:
\begin{itemize}[itemsep=0mm]
\item[] \textbf{Movie 1}: 2D simulations of an individual cell polarizing and migrating in the clockwise (top), counterclockwise (middle) directions, and non-coherent (bottom) rotational movement.
\end{itemize}

\clearpage


\appendix 

\input{suppmat}

\end{document}

%% file: abstract.tex
\begin{abstract}
\vspace{2mm}

Chiral patterns have been observed in various processes from swirling bacterial colonies to tissue morphogenesis and cytoskeletal organization, yet the physical mechanisms underlying chiral cell motion remain poorly understood. Motivated by experiments demonstrating directional bias in the circular motion of confined cells, we use the tools of dynamical systems analysis with computer simulations to identify minimal intrinsic and extrinsic mechanisms capable of generating persistent biased migration. The dynamical systems framework reveals a common organizing principle: directional bias emerges through changes in the stability and/or basins of attraction of the clockwise and counter-clockwise motility states. We find four distinct routes to such bias. First, intrinsic torque in a polarized cytoskeleton can be spatially integrated to produce biased circular motion. Second, anisotropic cell–substrate friction can generate directional preference when reduced friction along the polarity axis is coupled to a directional offset. Third, a chiral wall-alignment response can also produce a persistent directional preference. Finally, substrate patterns that break mirror symmetry, such as dextral or sinistral ridges and troughs, can likewise bias rotational direction. Together, these mechanisms yield distinct, testable predictions and suggest a unifying lens for experimental interrogation of cellular chirality and the design of synthetic systems with programmable chiral motion.
\end{abstract}
\vspace{5mm}

%% file: main_text.tex

\section{Introduction}

From embryonic development to cell migration, the establishment of spatial asymmetries is a fundamental biological process. One manifestation of such asymmetry is chirality --- the geometric property of an object that lacks mirror symmetry. Cellular chirality has been observed in diverse settings, including cell migration~\cite{Xu2007,Vecchio2024,Brangwynne2000,Guillamat2022,Fernandez2021,Badih2025,Lee2021,Hachem2024,Chin2018,Erzberger2020,Kozak2023,Lu2024}, alignment on micropatterns~\cite{Tee2023,Wan2011,Chen2012,Turiv2020,Endresen2021,Kaiyrbekov2023}, cytoskeletal organization~\cite{Liu2016,Tee2015,Tee2023}, cell morphology~\cite{Chin2018,Worley2018,Ray2018}, and multicellular tissues~\cite{Cetera2014,Founounou2021}. Although chirality is often linked to cytoskeletal organization, environmental factors such as confinement geometry~\cite{Mahmud2009}, substrate patterning~\cite{Turiv2020,Endresen2021,Kaiyrbekov2023}, and external fields~\cite{Moreau2019,Wang2020,Zhao2019,Wang2020,Renkawitz2019} can also bias migration. Yet the physical mechanisms that generate persistent directional bias remain poorly understood. Determining how chirality emerges at the single-cell level is therefore a necessary step toward explaining how left–right asymmetries arise across biological scales.

To investigate the physical origins of directional biases, we focus on a minimal model of the migration of a polarized active cell within a disk-shaped confinement. Exploiting the simplicity of the modeling framework, we combine computational simulations with dynamical systems analysis to identify the minimal intrinsic and extrinsic mechanisms capable of generating persistent chiral motion. Building on our previous work on confined cell doublets~\cite{Im2025}, we show that multiple symmetry-breaking mechanisms can produce robust directional bias, including intrinsic cytoskeletal chirality, anisotropic friction, substrate patterning, and chiral wall interactions. Despite their apparent differences, these mechanisms share a common design principle: persistent chiral motion emerges from an offset between the body-frame polarity axis and the lab-frame axis governing cell–substrate interactions. These mechanisms suggest concrete strategies for experimentally probing the origins of cellular chirality while informing the design of active and living systems with controllable directional behavior.

The rest of the paper is organized as follows. In Sec.~II we lay out the biophysical model and its assumptions. Throughout, we follow a single observable: the split of cells between clockwise (CW) and counter-clockwise (CCW) rotation, and we ask what makes that split even, biased, or fully one-sided. The answer depends on one property of the underlying dynamics: whether the cell circulates without ever settling (a conservative system) or settles into a preferred state (a dissipative system). In Sec.~III.A, the system is conservative and has no built-in handedness, giving an even CW/CCW split that no parameter can tip. In Sec.~III.B, an intrinsic torque in the protrusion dynamics gives the cell a directional preference but leaves the dynamics conservative, which tilts the CW/CCW split toward one side. Sections~III.C and~III.D make the dynamics dissipative by introducing a chiral offset into one of the two cues that steer the cell's polarity. These two cues are its protrusion direction and the inward direction due to the confining geometry; an anisotropic friction with the substrate underneath selects an ``easy" direction of migration, independent of where the cell starts. In Sec.~III.E, an adhesive spot in the middle of the confining geometry also makes the dynamics dissipative but without imposing handedness: the CW/CCW split stays even, but now a bistability emerges that can produce tunable chirality when coupled with another mechanism. Finally, in Sec.~III.F, we extend the anisotropic-friction mechanism to patterned substrates, where the preferred migration axis is fixed in the laboratory frame. Because these systems no longer admit the same low-dimensional description, we rely on simulations and show that only patterns lacking mirror symmetry generate directional bias. More broadly, our analysis reveals how diverse biological and physical mechanisms can be understood within a common dynamical framework, where chirality emerges through changes in the stability and accessibility of competing motility states.

\section{Generic description of an active particle in confinement}

We propose a minimal mathematical model to capture the interplay between cell polarity, contractility, and spatial confinement (Fig.~\ref{fig:singlet_schematic}). In our model, the cell is represented by the two-dimensional position of its center-of-mass (ignoring spinning), and further detailed morphology, such as the cell membrane, nucleus, and actomyosin cytoskeleton, is ignored.

\begin{figure*}[h]
    \centering
    \includegraphics[width=0.75\textwidth]{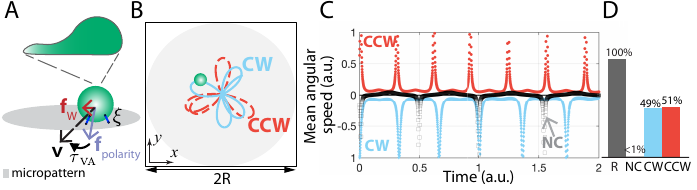}
    \caption{\textbf{Single migrating cell model and its dynamics.} (A) Singlet model schematic. (B) Sample CW and CCW trajectories over one rotational cycle. 
    (C) Angular velocity over~2 arbitrary time units for cells that move non-coherently by switching directionality (black), rotate CW (blue) and CCW (red). (D) Averaged number of rotating singlets and their distribution into the three states: CCW (51\%), CW (49\%), and NC ($<0.5$\%) out of 3200 model simulations.}
    \label{fig:singlet_schematic}
\end{figure*}

Taking the overdamped (or viscous dominated) approximation to a physical system, a cell at location $\mathbf{x}(t) = R(\cos\theta, \sin\theta)^T$ moves with velocity $\mathbf{v}(t)$ according to its local force balance:
\begin{equation}
    \sum \mathbf{f}(\mathbf{x}(t),t) = 0 \Rightarrow \mathbf{v} = \dot{\mathbf{x}} = \mathbb{M} \mathbf{f}_\text{polarity},
\label{eq:singlet_forcebalance}
\end{equation}
where $\mathbb{M}$ is the mobility matrix. In general, this matrix can capture anisotropies, hydrodynamic interactions, or spatial heterogeneities in the substrate. Here, however, we assume $\mathbb{M} = \frac{1}{\xi} \mathbb{I}$, implying that velocities are directly proportional to applied forces with $\xi$ as the viscous drag coefficient. This is equivalent to elastic cell-matrix interactions under conditions of high dissociation rates of these adhesion bonds~\cite{Howard2018}. The forces acting on the point are: (1) frictional drag force between the cell and the surface underneath ($\mathbf{f}_\text{drag} = \xi \dot{\mathbf{x}}$), and (2) active polarity force arising from front-rear signaling of actin-based protrusions and myosin-based contraction ($\mathbf{f}_\text{polarity}$).

\textbf{Cell polarity.} 
Migrating cells have an underlying chemical polarization, indicating the areas of the cell that are likely to protrude (``front-like'') and those likely to contract (``rear-like'')~\cite{LeahReview}. This can include asymmetric distribution of Rho GTPases, with Rac1 activity driving the cell front through lamellipodial extensions, and RhoA promoting myosin contractility in the rear. Rather than explicitly modeling the dynamics of one or more Rho GTPases~\cite{Mori2008,Buttenschon2020,CoposMogilner2020}, we summarize cell polarity with a single spatiotemporal motility force $\mathbf{f}_\text{polarity}$. This force is a vector with direction $\phi$, initially chosen randomly, and magnitude $\gamma_{\mathrm{pol}}$, a model parameter:
\begin{equation}
    \mathbf{f}_\text{polarity} = \gamma_{\mathrm{pol}} \hat{\mathbf{p}} =  \gamma_{\mathrm{pol}}\left(\cos{\phi},\sin{\phi}\right)^{T},
\label{eq:polarity_a}
\end{equation}
\noindent and over time evolves via the mechanism of velocity alignment~\cite{Camley2014}:
\begin{equation}
    \dot{\phi} = \frac{\abs{\vb{v}}}{V_0} \frac{\sin\lb \theta_V - \phi \rb}{\tau_{\rm VA}}.
\label{eq:polarity_b}
\end{equation}
\noindent Here, $\hat{\mathbf{p}}$ is the unit vector with direction $\phi$, $\theta_V$ is the angle of the polar parameterization of the velocity vector $\mathbf{v}$ and computed every time step as $\arctan{(v_y/v_x)}$, $\tau_\text{VA}$ sets the orientational persistence timescale, and $V_0 = \gamma_{\rm pol} / \xi$.
A long timescale, $\tau_\text{VA} \gg 1$, would imply the cell's orientation is insensitive to its own velocity direction and moves with high persistence in the initial polarity direction.
This form is in the same spirit as the one suggested originally by~\cite{Szabo2006} and later adapted by~\cite{Camley2014} and others.

\textbf{Spatial confinement.}
To ensure that the cell remains geometrically constrained to the fibronectin-coated micropattern, we incorporate confinement through a reorientation of the polarity angle in Eq.~\eqref{eq:polarity_b}:
\begin{equation}
    \dot{\phi} = \frac{\abs{\vb{v}}}{V_0} \frac{\sin\lb \theta_V - \phi \rb}{\tau_{\rm VA}}  
        + \frac{\sin\lb \theta_W - \phi \rb}{\tau_W}    
\label{eq:polarity_c}
\end{equation}
where $\theta_W = \pi + \theta$ is the angle pointing from the cell's position toward the center of the domain. 


This implicit formulation captures how cells with finite spatial extent continuously sense the domain boundaries through their cytoskeleton. Since the polarization represents an averaged cytoskeletal response, the reorientation mechanism naturally accounts for the cell's distributed sensing of geometric constraints without requiring an artificial radial potential that switches on and off at a specific threshold. In the absence of such a confinement effect, the cells engage in persistent directional motion (Fig.~S1A).

\textbf{Quantification of rotational motion by the winding number.} To quantify the rotational behavior of trajectories, we use the position angle $\theta(t)$ 
and its instantaneous angular velocity $\omega(t) \equiv \dot\theta$. 
For a trajectory with $R(t) > 0$ on $[0, T]$, the \emph{winding number} is
\begin{equation}\label{eq:Theta_int}
    n(T) = \frac{1}{2\pi} \int_0^T \omega(t)\,dt = \frac{1}{2\pi}\lb \theta(T) - \theta(0) \rb.
\end{equation}
To accurately compute cumulative rotation, angular increments are unwrapped to account for the periodicity of $\theta \in [-\pi,\pi]$, ensuring that full revolutions are retained in the winding number calculation. 
A trajectory has net CCW rotation if $n(T) > 0$, net CW rotation if
$n(T) < 0$, and no net rotation if $n(T) = 0$. 



\section{Results}

\subsection*{Movement in confinement produces a conserved system with reflection-symmetry}

To probe the system response, we place a single cell in the center of the disk-shaped micropattern, induce polarization in a random direction, and allow the system to evolve according to Eqs.~\eqref{eq:singlet_forcebalance}, \eqref{eq:polarity_a}, and~\eqref{eq:polarity_c} for over 3 full rotation cycles (Fig.~\ref{fig:singlet_schematic}A). Fig.~\ref{fig:singlet_schematic}B plots one rotational cycle, in either direction, as the singlet navigates the confining geometry in our simulation. Notably, we find that the motion of our default cell displays a petal-like pattern unlike the circular paths observed with cell doublets by the same model~\cite{Im2025}.
Along with trajectories, we plot the angular speed for several initializations and the averaged behavior over 3,200 simulations.\footnote{ Without making assumptions about the underlying ratio of a binary proportion ($p=0.5$), and assuming a margin of error $E = \pm3\%$, we require $n = (Z^2\times p \times (1-p)) / E^2 = 1,068$ samples for $95\%$ confidence ($Z=1.96$). Assuming that the lowest ratio of rotation to non-coherent movement never goes below $30\%$, we conclude that $3,200$ samples yield $95\%$ confidence with error intervals of $\pm 3\%$. Any cases where the ratio of rotation to non-coherent movement falls below $40\%$ will be classified as statistically unreliable.}
(Fig.~\ref{fig:singlet_schematic}C-D). Equations of motion are solved numerically using the Forward Euler integration scheme, and parameter values are provided in Table~S1.

\begin{figure*}[h!]
    \centering
    \includegraphics[width=0.65\textwidth]{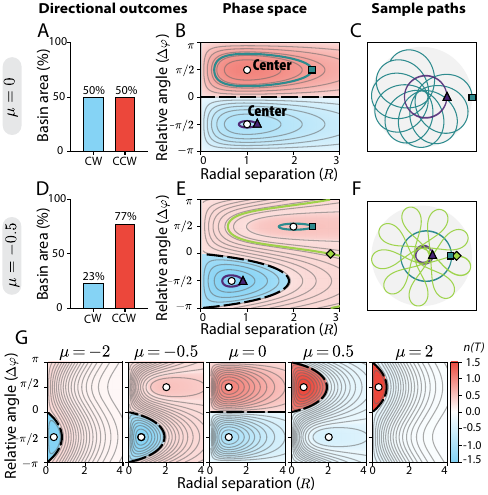}
    \caption{\textbf{Directionality bias can either persist or disappear in confining geometries when cells have an added intrinsic torque in their polarity orientation.} (A, D) CW/CCW motility outcomes for $\mu=0$ and $\mu=-0.5$. (B, E) Phase portraits in $(R,\Dphi)$ colored by the time-averaged winding number $n(T)$, with representative trajectories (gray and colored lines), flow direction (arrows), invariant lines (dashed black), and neutral centers (white circles). (C, F) Example cell trajectories corresponding to selected initial conditions in (B, E). (G) Phase space scanning the intrinsic bias parameter $\mu$ from $-2$ to $2$, colored by $n(T)$, shows the breakdown of CW/CCW rotational states.}
    \label{fig:singlet_results}
\end{figure*}

Cells display a variety of motility patterns (Fig.~S2), including cells moving persistently in one direction, either clockwise (CW) or counterclockwise (CCW), or switching direction (Fig.~\ref{fig:singlet_schematic}C-D, Movie 1), thus not moving coherently (NC). The marginal case of non-coherent movement appears at a very low frequency. To classify the emergent behavior, 
winding number is used -- non-coherent movement occurs if the cell switches directionality over some arbitrary percentage.\footnote{Over 20\% unidirectional is considered coherent rotation.}, with clockwise for negative values or counterclockwise for positive values of the winding number. This means that our model, at its default parameters, predicts three-state behavior, with nearly all cells persistently rotating either clockwise or counterclockwise with equal transition rates (Fig.~\ref{fig:singlet_schematic}E, Fig.~\ref{fig:singlet_results}A).

To verify the computational simulation results, we transform to polar coordinates where $\mathbf{x} = R \left( \cos\theta, \sin\theta \right)^T$ and introduce the phase-lag $\Dphi = \phi - \theta$, quantifying the angular offset between polarity and positional angle. 
From the equation of motion, we find that $\mathbf{v} = \gamma_\text{pol} \hat{\mathbf{p}}/\xi$, and consequently the velocity is parallel to the polarity, so that $\theta_V = \phi$. The polarity equation reduces to $\dot{\phi} = \sin(\Delta \phi)/\tau_W$ since $\sin(\theta_W-\phi) = \sin(\theta + \pi - \phi) = \sin(\pi-\Delta \phi) = \sin \Delta \phi$. By re-writing the equation of motion in polar coordinates (Appendix A), and subtracting $\dot\theta$ from $\dot\phi$, one can derive the closed system on $(R, \Dphi)$:
\begin{equation}\label{eq:sym_reduced_nobias}
    \dot R = \frac{\gamma_\text{pol}}{\xi}\cos\Dphi,
    \quad
    \dot{\Dphi} = \sin\Dphi\lb \frac{1}{\tau_W} - \frac{\gamma_\text{pol}}{\xi R}\rb.
\end{equation} 

Phase plane analysis of this reduced system explains the equal CW/CCW split (Fig.~\ref{fig:singlet_results}B, Appendix B). The phase space divides into two regions, an upper half in which the cell migrates counter-clockwise and a lower half in which it migrates in the clockwise direction, set apart by invariant dividing lines that the dynamics never crosses. Therefore, a cell retains whichever direction it starts with, and each region is organized around a single center, the only state in which $R$ does not oscillate (Fig.~\ref{fig:singlet_results}C). 

As suggested by the closed trajectories of the reduced 2D plane, we find that the system admits a conserved Hamiltonian quantity, $H(R,\Dphi)$ 
(Appendix B). Trajectories lie on the level sets of $H$ (Fig.~\ref{fig:singlet_results}B). Moreover, the reduced system is invariant under reflection; every trajectory $(R,\Dphi)$ has a mirror image trajectory $(R,-\Dphi)$ (Appendix B). Since the model parameters are scalar quantities that are unchanged by this reflection, no parameter variation can break the equal CW/CCW split. We confirmed this result numerically with parameter sweeps (Fig.~S1B). 

\subsection*{Intrinsic torque breaks direction symmetry}

In pursuit of symmetry breaking in the system, and emergence of directionality preference, we probed the model's response to the presence of intrinsic bias in the polarity machinery. We extended the polarity model in Eq.~\eqref{eq:polarity_c} 
by adding an intrinsic torque $\mu$:
\begin{equation}
    \dot{\phi} =
        \frac{1}{\tau_\text{VA}} \frac{\abs{\vb{v}}}{V_0} \lsb 
            \sin\lb\theta_V - \phi \rb + \mu             
        \rsb 
    + \frac{1}{\tau_\text{W}}
        \sin\lb\theta_W - \phi\rb.   
\label{eq:polarity_d}
\end{equation}
Here, $\mu < 0$ leads to CW rotations, and $\mu > 0$ to CCW rotations.
For example, for $\mu=-0.5$, this intrinsic bias in the evolution of the polarity angle can translate into a bias towards CW rotational movement (Fig.~\ref{fig:singlet_results}D, trajectories in Fig.~S3). 

Phase plane analysis confirms the model simulations: for $\mu\neq0$ the direction symmetry is broken (Fig.~\ref{fig:singlet_results}D-E, Appendix C). Importantly, the bias does this without destroying the conserved structure of the unbiased model: the reduced system still admits
a conserved quantity 
so the two rotational states remain neutrally stable centers. What the intrinsic bias does change is the balance between the two regions: each center persists but shifts, so that for $\mu < 0$ the clockwise center tightens into a smaller orbit (Fig.~\ref{fig:singlet_results}F) and claims a larger share of initial conditions Fig.~\ref{fig:singlet_results}G), while counter-clockwise center widens and claims fewer initial conditions. Every trajectory retains a definite, signed winding number, clockwise or counter-clockwise.

As the intrinsic bias grows further, it pushes the disfavored center into ever-wider orbits, until a critical strength ($\abs{\mu} = 1/\tau_W$), where its radius diverges and it collides with a saddle point sitting at infinity, annihilating in a saddle-center bifurcation (see Appendix C). Beyond this point 
a single rotational state survives, CCW for $\mu > 1/\tau_W$ and CW for $\mu < -1/\tau_W$, and the cell turns one way regardless of how it starts. An intrinsic bias ($\mu \neq 0$) is therefore both necessary and sufficient to produce directed rotation in a single cell. In previous work~\cite{Im2025}, we found that this type of dynamics, extended to cell pairs, can produce richer behavior, including regimes where doublets overcome an intrinsic clockwise bias to rotate predominantly counter-clockwise.

\subsection*{Lower frictional drag along an offset direction creates directional bias}

Above, we explored the possibility that a directional preference in cellular motility can arise from an intrinsic torque in the cellular response (polarity), but how that can emerge remains opaque. To illustrate one possible underlying mechanism, we allow cell-substrate frictional drag to be anisotropic: moving in one direction is easier than moving across it. The anisotropy could be introduced through several distinct possibilities, including intrinsic cytoskeletal alignment, confinement-induced shape anisotropy, or patterning of ridges and troughs directly onto the surface (see Section III.F). We show that a frictional anisotropy in the model can produce a preferred migratory direction through a mechanism that is dissipative rather than the conserved dynamical system in the previous section (Fig.~\ref{fig:frictionanisotropy_results1}).

\begin{figure*}[h!]
    \includegraphics[width=0.95\textwidth]{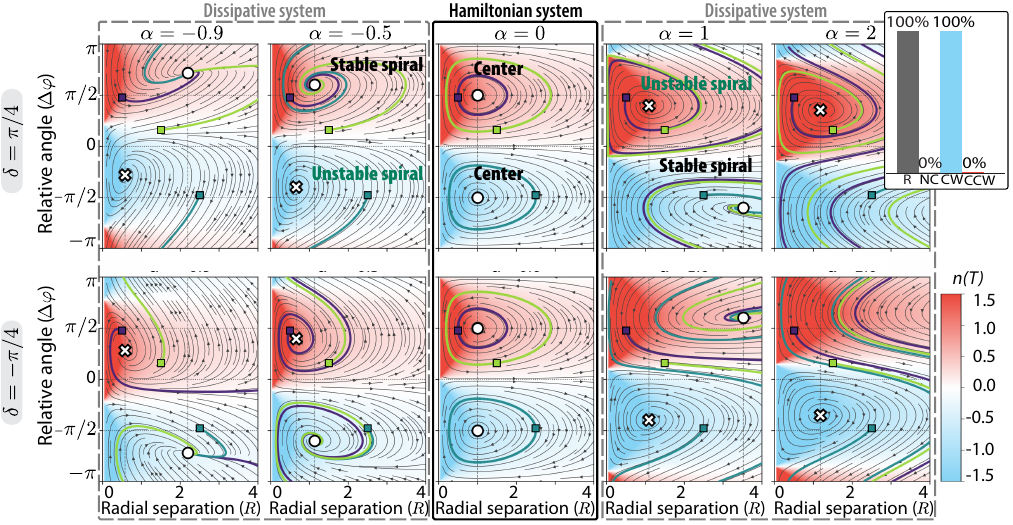}
    \caption{\textbf{Phase portraits illustrate the transition from conserved (Hamiltonian) to dissipative dynamics} and the resulting changes in the stability of CW and CCW equilibria with offset angle $\delta$ and frictional anisotropy $\alpha$. Top: $\delta=\pi/4$; bottom: $\delta=-\pi/4$. From left to right, $\alpha$ increases from negative to positive values. Phase space $(R,\Dphi)$ is colored by the time-averaged winding number $n(T)$. White circles denote stable (or neutrally stable) equilibria, white crosses unstable equilibria, gray curves representative trajectories, and arrows the flow direction. Inset: fraction of rotating singlets and their distribution among motility states for $\delta=\pi/4$ and $\alpha>0$.}
    \label{fig:frictionanisotropy_results1}
\end{figure*}

First, the intrinsic bias is removed from Eq.~\eqref{eq:polarity_d}. To model anisotropic friction, we introduce an ``easy'' axis, $\hat{\mathbf{e}}_\parallel$; we assume that motion parallel to this easy axis is governed by the friction coefficient $\xi_\parallel$, and $\xi_\perp$ in the perpendicular direction. The ratio of these coefficients controls migratory persistence anisotropy ($\xi_\perp>\xi_\parallel$) by promoting alignment with the easy direction. 
The mobility matrix tensor is 
\begin{equation}
    \mathbb{M} = \frac{1}{\xi_\perp} \lb I + \alpha \hat{\vb{e}}_{\parallel} \hat{\vb{e}}_{\parallel}^{T} \rb, 
    \quad 
    \alpha = \frac{\xi_\perp - \xi_\parallel}{\xi_\parallel} > -1.
\label{eq:mobility_defn}
\end{equation} 
The $\alpha$ term introduces an orientation-dependent coupling between radial and angular dynamics. 
In the limit that $\alpha=0$ ($\xi = \xi_\perp=\xi_\parallel$), Eq.~\eqref{eq:mobility_defn} reduces to $\mathbb{M} = \frac{1}{\xi} \mathbb{I}$ as expected. Note that for $\alpha \leq -1$ the mobility matrix is no longer positive-definite and we exclude this case.

If we assume that easy axis is lagging the polarity by a fixed offset angle $\delta$, $\hat{\mathbf{e}}_\parallel (\phi) = \left(\cos(\phi+\delta),\sin(\phi+\delta)\right)^{T}$, one can show that the update velocity becomes
\begin{equation}
    \mathbf{v} = \frac{\gamma_\text{pol}}{\xi_\perp} 
        \left[(1+\alpha \cos^2\delta)\hat{\mathbf{p}} + \frac{\alpha}{2}\sin(2\delta)\hat{\mathbf{p}}_\perp\right].
\label{eq:sectionc_eom}
\end{equation}
Here, $\hat{\mathbf{p}} = \mathbf{f}_\text{polarity}/\gamma_\text{pol} = (\cos\phi,\sin\phi)^{T}$ as in Eq.~\eqref{eq:polarity_a}, and $\hat{\mathbf{p}}_\perp = (-\sin\phi,\cos\phi)^{T}$ (Appendix D). We find that the cell’s velocity is not perfectly aligned with its polarity vector due to the presence of the second term in Eq.~\eqref{eq:sectionc_eom}. Instead, the cell picks up a small sideways component (proportional to $\sin(2\delta)$), meaning it ``slips'' slightly sideways as it moves. Consequently, the velocity alignment term 
does not vanish and remains as a constant applied torque to the system: 
$    \dot{\Dphi} = \frac{1}{\tau_W} \sin\lb \theta_W - \phi \rb + \mu_\text{eff}$,
where 
$    \mu_\text{eff} = \alpha \sin(2\delta)/2\tau_\text{VA}$.

For $\delta = 0$, the system simplifies to the unbiased case, with no intrinsic directional preference, as in Eq.~\eqref{eq:sym_reduced_nobias}. However, a nonzero offset ($\delta \neq 0$) does more than add a bias; it changes
the dynamics qualitatively. 
The offset destroys the Hamiltonian dynamics and makes the reduced $(R, \Dphi)$ flow non-conservative. For moderate anisotropy the two neutrally stable centers become hyperbolic spirals, one stable and one unstable; for stronger anisotropy the sink instead escapes to infinity, a nonphysical runaway we set aside
(Fig.~\ref{fig:frictionanisotropy_results1}, Appendix D). These equilibria are spirals, never saddles ($\det J > 0$), and their stability is set entirely by the 
sign of $\alpha\sin(2 \delta)$. 
We find that the resulting bias produces all-or-nothing CW/CCW outcomes rather than a tilting of outcome probabilities; for example, if $\alpha > 0$ (favored alignment with the easy axis), a positive $\delta$ stabilizes the clockwise state, and destabilizes the counter-clockwise one (top right, Fig.~\ref{fig:frictionanisotropy_results1}).

Because the stable spiral is a global attractor for moderate anisotropy, this directional preference is robust to initial conditions: every initialization of the singlet converges to the same rotational state, which our computational simulations
confirm (all tested initial conditions produce persistent clockwise rotation; inset Fig.~\ref{fig:frictionanisotropy_results1}).
Beyond the threshold anisotropy the interior attractor is lost and the reduced trajectory instead spirals outward to unbounded
radius (Appendix~D). Such an escape is unphysical for a cell confined to the disc, so we do not consider it further.
The dissipation also reshapes the orbits, replacing the closed petal-like trajectories of the conservative model with 
spirals that settle onto a singular circular orbit (Fig.~S4).

Taken together, these results demonstrate that reduced cell-substrate frictional resistance along the polarity axis, combined with an offset arising from chiral cytoskeletal organization, is sufficient to generate a directional bias in cell movement. The next mechanism breaks symmetry in the same dissipative way, but through the cell's interaction with the confining geometry.

\subsection*{Chiral alignment due to confinement selects a single motility direction}

Another mechanism that breaks the reflection symmetry of Eqs.~\eqref{eq:singlet_forcebalance}-\eqref{eq:polarity_c} is through confinement effects. Here, we return to isotropic cell-substrate friction, so the cell's velocity is again parallel to the polarity ($\theta_V = \phi$) and the velocity alignment term drops out; the symmetry breaking comes entirely from confinement. We offset the alignment of polarity direction by a fixed angle $\chi$, replacing $\sin\lb\theta_W - \phi\rb$ with $\sin\lb\theta_W - \phi + \chi\rb$. Together with $\theta_W = \theta + \pi$, the confinement term becomes $\sin\lb\Dphi - \chi\rb/\tau_W$, and the closed system on $(R, \Dphi)$ is
\begin{equation}
    \dot R = \frac{\gamma_{\rm pol}}{\xi}\cos\Dphi, 
    \quad 
    \dot{\Dphi} = \frac{\sin\lb\Dphi -\chi\rb}{\tau_W} - \frac{\gamma_{\rm pol}}{\xi R}\sin\Dphi.
    \label{eq:sym_reduced_chiralwall}
\end{equation}
Notably, the symmetric flow in Eq.~\eqref{eq:sym_reduced_nobias} has shifted by $\chi$. 

Expanding the chiral term in Eq.~\eqref{eq:sym_reduced_chiralwall} gives
$\frac{1}{\tau_\text{W}} \sin(\Dphi - \chi) = 
        \frac{1}{\tau_\text{W}} \cos\chi\sin\Dphi - \frac{1}{\tau_\text{W}} \sin \chi \cos\Dphi$.
The second term in the expression breaks reflection symmetry ($\Dphi \rightarrow -\Dphi$), but unlike the constant torque $\mu$, it carries a $\cos\Dphi$ factor, so it vanishes at the two equilibria $\Dphi = \pm\pi/2$. This makes it tempting to conclude that $\chi$ produces no bias, and indeed $\chi$ does not shift the steady states. Instead, exactly as for the offset from the polarity axis, it renders the system dissipative and splits the stability of the two states: centers become hyperbolic spirals, one stable and one unstable, at a common rescaled radius (Appendix E, Fig.~S5). The stability of the equilibria, and hence the handedness of the motility, is dictated by the sign of parameter $\chi$. Notably, chiral confinement produces the same all-or-nothing bias as the polarity axis offset, by splitting the stability of the two rotating states rather than displacing them.

\subsection*{A center adhesive spot turns intrinsic chirality mechanisms into tunable bias}

In an effort Ferencto introduce different cell-matrix adhesion topology, we first consider a radial cell-matrix force at the center of the disk-shaped pattern (Fig.~\ref{fig:frictionanisotropy_results2}A, pattern (ii)) in Eq.~\eqref{eq:singlet_forcebalance}:
$\vb{f}= \vb{f}_\text{polarity} + \vb{f}_\text{anchor} = \gamma_{\rm pol} \hat{\vb{p}} + f(R)\hat{\vb{r}},$ where $f(R) = k_0 R$. 
The force is radially oriented and is equivalent to an elastic spring tethering the cell to the origin; by controlling the sign of its coefficient $k_0$, the force can be either repelling ($k_0<0$) or attracting ($k_0>0$). The reduced dynamics are given by
\begin{equation}
\begin{split}
    \dot R &= \frac{\gamma_{\rm pol}}{\xi}\cos\Dphi + \frac{f(R)}{\xi}, \\
    \dot{\Dphi} &= \sin\Dphi\lsb \frac{1}{\tau_W} - \frac{\gamma_{\rm pol}}{\xi R} - \frac{f(R)}{\gamma_{\rm pol}\tau_{\rm VA}} \rsb.
\end{split}
\end{equation}

In the absence of the additional radial force, Eq.~\eqref{eq:sym_reduced_nobias} is recovered. For $f \neq 0$, the cell's velocity is no longer parallel to its polarity direction, so, as with the polarity offset, the flow becomes dissipative and the neutrally stable centers become spirals. 
Yet, similar to the default cell on a disk-shaped micropattern, the outcomes remain unbiased, with an equal likelihood of CW and CCW motion, as evidenced by the equally sized, well-defined basins of attraction of the two steady states  (Fig.~\ref{fig:frictionanisotropy_results2}D(i)). We conclude that the presence of an adhesive dot alone does not produce bias in the directional motion of the singlet cells exploring the confined geometry.

\begin{figure*}[htbp]
    \centering
    \includegraphics[width=0.5\textwidth]{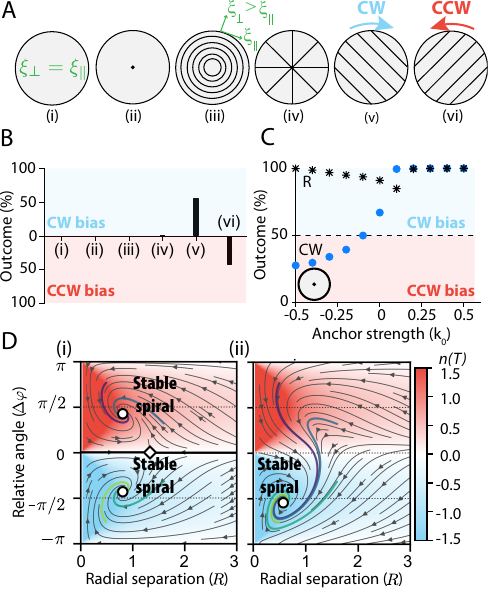}
    \caption{\textbf{Model predicts directional biases on substrate patterns that lack mirror symmetry.} (A) Substrate adhesive patterns. (B) Directionality outcomes showing deviations from equal CW/CCW distribution for substrate patterns in (A). (C) Motility outcomes of the CW-biased cell ($\mu=-0.5$) on pattern (ii) as a function of anchor strength $k_0$. (D) Phase portraits in $(R,\Dphi)$ for pattern (ii) with a center adhesive dot ($k_0=0.5$, $k_W=1$), for (i) an unbiased cell ($\mu=0$) and (ii) a CW-biased cell ($\mu=-0.5$), respectively.}
    \label{fig:frictionanisotropy_results2}
\end{figure*}


What the presence of the additional radial force enables is a tunable bias in the chirality mechanisms we have discovered so far. Because the two equilibria are not attractors, any symmetry breaking coupling (an intrinsic torque $\mu\neq0$, a polarity axis offset, $\delta\neq0$, or a chiral confinement $\chi\neq0$) tilts the basin boundary toward the disfavored sink and enlarges the favored basin, so the population handedness shifts continuously with the symmetry breaking strength. For example, for the case of a CW intrinsically biased cell ($\mu<0$), changing the directionality of the anchor (attractive or repulsive) reveals either an overall CW or a CCW directional bias in both the phase space analysis and model simulations (Fig.~\ref{fig:frictionanisotropy_results2}B-D). With an attracting anchor, CCW bias emerges, while a repulsive anchor enhances the intrinsic CW bias. 

Dynamical system analysis further shows that a strong enough intrinsic bias ($\mu$) or chiral confinement ($\chi$) removes
the disfavored state and makes the cell one-handed, by one of two routes: the intrinsic bias slides the disfavored sink 
into the basin boundary until the two annihilate (a saddle-node bifurcation), whereas the chiral confinement drives the 
disfavored sink unstable (a Hopf bifurcation) (see Appendix~F, Fig.~S6). The polarity axis offset (anisotropic friction)
is the exception: within the physical range its effective torque is bounded and too weak to remove a sink, so the 
$(\alpha, k_0)$ plane stays bistable and only skews the population. This behavior is illustrated in Fig.~\ref{fig:two_par}:
under some conditions, cells can rotate in either direction, but sufficiently strong intrinsic bias ($\mu$) or chiral
cell-wall interactions ($\chi$) ultimately favor a single direction of rotation.

\begin{figure*}
    \centering
    \includegraphics[width=0.5\linewidth]{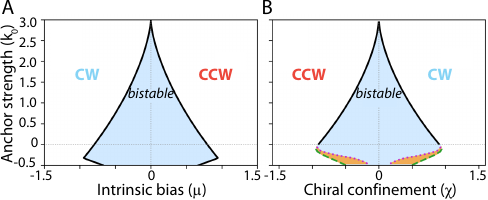}
    \caption{\textbf{Two-parameter bifurcation diagrams for the single cell on a substrate with a center adhesive spot} with (A) intrinsic bias ($\mu$) and (B) chiral cell-wall interactions ($\chi$). In both cases, there is a range of conditions where clockwise (CW) and counterclockwise (CCW) rotation can both occur. As the strength of the intrinsic bias or chiral confinement increases, the system transitions to a regime in which a single rotational direction is preferred.}
    \label{fig:two_par}
\end{figure*}

\subsection*{Substrate patterning as an extrinsic route to directional biases}

A final mechanism we report on for attaining a directional bias in confined motility is through patterned substrates. We focus on patterns as depicted in Fig.~\ref{fig:frictionanisotropy_results2}A and compare the CW/CCW split outcomes against the default disk-shaped isotropic pattern (Fig.~\ref{fig:frictionanisotropy_results2}B). In contrast to the previous sections, these patterns do not generally reduce to the same low-dimensional (2D) dynamical system and therefore are not amenable to the same analytical treatment. Nevertheless, their behavior can still be systematically characterized through numerical simulations. 

To model a patterned substrate, we assume that there is spatial variation along the easy axis, $\mathbf{\hat{e}}_\parallel = g \mathbf{\hat{x}} + h \mathbf{\hat{y}}$, where $g, h$ are components of $\mathbf{\hat{e}}_\parallel$ in the lab frame such that $g^2+h^2=1$. A simple computation in 2D gives
$\mathbb{M} = \frac{1}{\xi_\perp} \lb I + \alpha \begin{pmatrix} h^2 & -gh \\ -gh & g^2\end{pmatrix} \rb,$
where $\alpha = \frac{\xi_\perp - \xi_\parallel}{\xi_\parallel}>-1$ is defined as before in Eq.~\eqref{eq:mobility_defn}. In Fig.~\ref{fig:frictionanisotropy_results2}A, substrate patterns (iii) and (iv) are defined by $g = g(\vb{r}) = \cos(\theta+\theta_0)\hat{\vb{x}}$ and $h = h(\vb{r}) = \sin(\theta+\theta_0)\hat{\vb{x}}$ with $\theta_0 = \pi/2$ and $0$, respectively. Patterns (v) and (vi) utilize $g = 1/\sqrt{2}\, \hat{\vb{x}}$ and $h = \mp 1/\sqrt{2} \,\hat{\vb{y}}$, respectively. A summary of the averaged singlet motility outcomes is provided in Fig.~\ref{fig:frictionanisotropy_results2}B. 

We find that the even distribution in motility direction is disrupted only by patterns that lack mirror symmetry: patterns (v) and (vi) in Fig.~\ref{fig:frictionanisotropy_results2}A. For the case of mirror symmetric patterns (Fig.~\ref{fig:frictionanisotropy_results2}A(iii)-(iv)), further varying $\alpha$ does modulate the number of rotational outcomes, but does not impact the directionality bias (Fig.~\ref{fig:frictionanisotropy_results2}B). For substrate patterns that lack mirror symmetry, cases (v) and (vi), the model predicts the emergence of a directional bias provided that the frictional anisotropy is sufficiently strong ($\alpha$, Fig.~S7A), even in the absence of any additional symmetry-breaking mechanisms. The resulting trajectories retain their characteristic petal-like structure but become increasingly skewed and elongated along the easy frictional direction (Fig.~S7B-C). This is because the lower frictional coefficient makes it easier to travel along that direction. As the degree of frictional anisotropy increases, the magnitude of the directional bias grows; however, the probability of observing persistent migration correspondingly decreases (Fig.~S7A). Instead, the cells increasingly switch directionality. For the results shown in Fig.~\ref{fig:frictionanisotropy_results2}B, we selected $\alpha=0.4$ as a representative value that produces a measurable directional bias while maintaining sufficient persistence. Because persistent trajectories become less frequent under these conditions, the number of simulation realizations was increased four-fold to obtain robust statistics. More broadly, the results suggest that spatial patterning of the cellular microenvironment may provide a simple extrinsic mechanism for guiding chiral migration and establishing directional preference in confined cells.

\section{Discussion}

Chirality is observed across biological scales, yet it is unknown how microscopic asymmetries converted into robust, macroscopic directional behavior? Here, we identify four minimal routes through which persistent chiral migration can emerge (Fig.~\ref{fig:hypothesis_attempt}): (1) an intrinsic torque manifested in the front-rear polarity establishment, (2) a frictional anisotropy along the protrusion direction (with an offset), (3) an intrinsic torque manifested due to geometric confinement, and (4) a frictional anisotropy along physical patterns on the substrate. Although these mechanisms originate from distinct biological or environmental sources, they share a common principle: directional bias arises when local asymmetries are coupled to a reference frame that breaks mirror symmetry. 

\begin{figure*}[h!]
    \centering
    \includegraphics[width=0.75\textwidth]{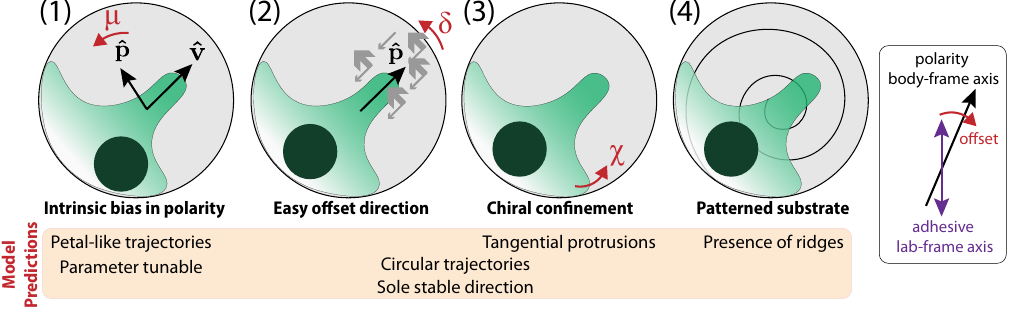}
    \caption{\textbf{Summary:} Identified mechanisms for chiral rotational movement of singlets in confining micropatterns: (1) intrinsic torque $\mu$ in polarity direction, (2) offset from the polarity direction through an ``easy" frictional direction, (3) chiral confinement forces, and (4) chiral physical patterning on the substrate. Model suggests that chirality in directional individual cell migration emerges with an offset between polarity axis and cell-substrate adhesion axis. 
    }
    \label{fig:hypothesis_attempt}
\end{figure*}

These identified mechanisms have been reported biologically. Experiments in dHL60 cells exposed to a uniform chemoattractant found a strong left-right bias in spontaneous polarization relative to the centrosome~\cite{Xu2007}. Because no external directional cue was present, these results suggest an intrinsic chiral asymmetry in the polarity machinery. In our framework, such asymmetry is represented by either an intrinsic torque or an offset angle between the polarity and motility axes, which acts as a symmetry-breaking term and biases the system toward persistent CW or CCW motion. Potential biological origins of such an intrinsic offset include persistent cytoskeletal anisotropy associated with mechanical memory~\cite{Gomez2026} and asymmetries generated by nuclear positioning within the cell~\cite{Renkawitz2019}. Consistent with our third mechanism for chiral motion, recent work suggests that intrinsically chiral actin organization can bias the orientation of cells relative to tissue boundaries, causing boundary cells to adopt a preferred left- or right-tilted alignment~\cite{Shi2025ordering}. This framing is congruent with our proposed chiral wall-alignment mechanism: the boundary acts as a reference frame, cells align to it with a handed offset, and that local interaction is sufficient to produce persistent large-scale chirality. Lastly, spatially constrained experiments of cells placed on ridge patterned surfaces~\cite{Kim2009,Ray2017,Su2023,Dong2019,Turiv2020,Endresen2021,Kaiyrbekov2023} support our fourth and final mechanisms for chiral motion.

A recurring theme across our results is that chirality is not solely an intrinsic property of the cell. Instead, environmental interactions can reveal, amplify, suppress, or even reverse an underlying directional preference. Fixed adhesion sites, geometric constraints, neighboring cells, or substrate architectures can serve as reference frames that transform latent asymmetries into observable migratory biases. This observation suggests that experimentally measured chirality may reflect an interplay between cellular properties and environmental context rather than either factor alone.

Accompanying computational simulations, the analytical framework reveals a common dynamical organization underlying chiral migration. Across multiple scenarios, the system exhibits competing CW and CCW attractors whose relative stability is tuned by intrinsic and extrinsic symmetry breaking effects. Different mechanisms act by reshaping this dynamical landscape, altering basin sizes, destabilizing equilibria, or transforming conservative dynamics into dissipative ones. Viewed through this lens, diverse manifestations of chirality can be understood as different perturbations of a shared underlying dynamical structure.


Our work provides a unifying framework for understanding how chirality can emerge from the interaction of cellular polarity, mechanics, and environmental structure. By identifying minimal routes to directional bias and the signatures associated with each, the framework offers experimentally testable hypotheses for probing the origins of biological chirality. At the same time, the principles uncovered here may inform the design of synthetic active materials and engineered systems, such as microvessels, that exploit symmetry breaking to achieve controlled chirality.

%% file: suppmat.tex
\begin{center}%
  {\LARGE Supporting Material for ``Directional bias of a single polarized cell under confinement''\par}%
  \vskip 1.5em%
  {\large
    \lineskip .5em%
    \begin{tabular}[t]{c}%
      Andreas Buttensch\"{o}n [a] \and Calina Copos [b,c]%
    \end{tabular}\par}%
  \vskip 1em%
\end{center}%
\par
\vskip 1.5em%
\setcounter{page}{1}

\setcounter{figure}{0}
\setcounter{table}{0}
\renewcommand{\thefigure}{S\arabic{figure}}
\renewcommand{\thetable}{S\arabic{table}}

\begin{table*}[!h]
\begin{center}
\begin{tabular}{llc}
\textbf{Parameter}       & \textbf{Description}       & \textbf{Value (a.u.) } \\ \hline
$r$ & Micropattern radius & 0.5 \\
$\Delta t$ & Simulation time step & 0.001 \\
$T_\text{max}$ & Total simulation time & 2 \\
\underline{\textit{Cellular parameters}} & & \\
$\mathbf{x_0} = (x_0,y_0)$ & Initial position of cell & $(0.01,0)$ \\
$\xi$ & Friction coefficient & 1.0 \\
$\gamma_\text{pol}$ & Polarity strength & 1.0 \\
\underline{\textit{Cell confinement parameters}} & & \\
$\tau_\text{W}$ & Confinement timescale & 0.0455 \\
\underline{\textit{Polarity parameters}} & & \\
$\mu$ & Intrinsic cellular bias & $-0.5$ (CW), $0.5$ (CCW) \\
$\tau_\text{VA}$ & Velocity alignment timescale & 0.5 
\end{tabular}
\end{center}
\caption{Table of parameter descriptions along with the values used in computer simulations and analysis.}
\label{Tab:modelparams}
\end{table*}

\begin{table}[H]
\centering
\small
\renewcommand{\arraystretch}{1.3}
\begin{tabularx}{\linewidth}{|p{2cm}|p{0.3cm}|p{1.6cm}|X|}
\hline
\textbf{Mechanism} & \textbf{SI} & $\mathbb{Z}_2$ broken? & \textbf{Phase portrait \& Ensemble outcome} \\
\hline
None & B & No & Two neutrally stable centers at $(\ell, \pm\pi/2)$; conserved Hamiltonian; equal $50/50$ CW/CCW. \\ \hline
Intrinsic torque & C & Yes & Centers displaced in radius; one escapes to infinity for $|\mu| > 1/\tau_W$; tunable bias. \\ \hline
Anisotropic friction & D & Yes & With moderate anisotropy centers become one spiral sink and one source spiral; single global attractor (one-handed bias). \\ \hline
Chiral wall & E & Yes & Centers become one spiral sink and one source spiral; single global attractor (one-handed bias). \\ \hline
Anchor & F & No & Centers become two spiral sinks plus axis saddles; globally bistable; equal $50/50$ CW/CCW. \\ \hline
\quad Anchor + 1 & F & Yes & Basins of attraction tilt; the disfavored sink is removed by a saddle-node (intrinsic bias or anisotropic friction) or a Hopf (chiral confinement) bifurcation. \\
\hline
\end{tabularx}
\caption{Summary of the dynamical systems analysis results for the singlet. Each mechanism of the fundamental system Eqs.~\eqref{eq:gen_master} is switched on alone; $\mathbb{Z}_2$ denotes the reflection symmetry $\Dphi \to -\Dphi$ of Appendix~B, and $\ell = V_0\tau_W$.}
\label{tab:roadmap}
\end{table}

\begin{figure*}[h!]
    \centering
    \includegraphics[width=\textwidth]{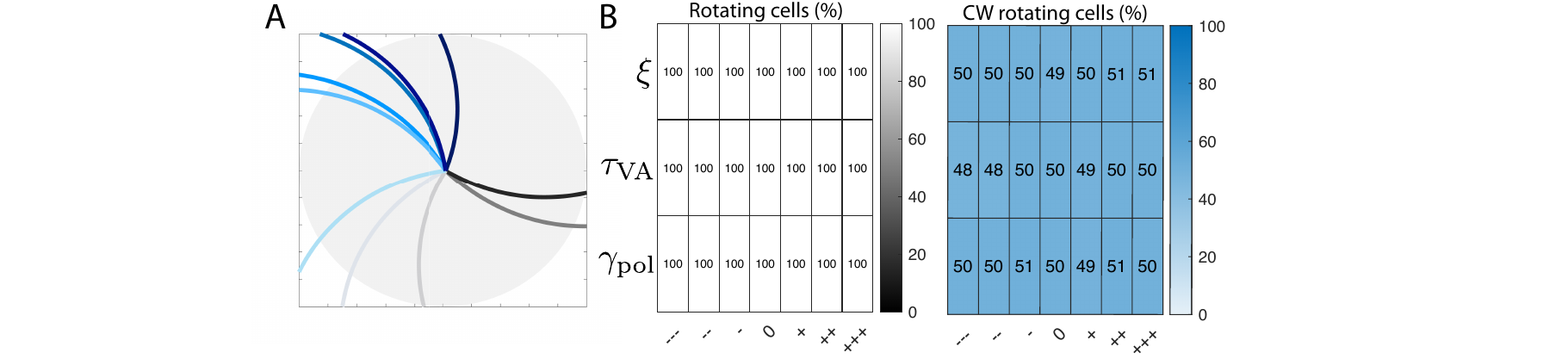}
    \caption{(A) Sample trajectories in the absence of confinement forces. (B) Heatmaps of the outcomes for the percentage of coherent rotations and CW rotations with parameter variations in the default unbiased singlet model.}
    \label{}
\end{figure*}

\begin{figure*}[h!]
    \centering
    \includegraphics[width=\textwidth]{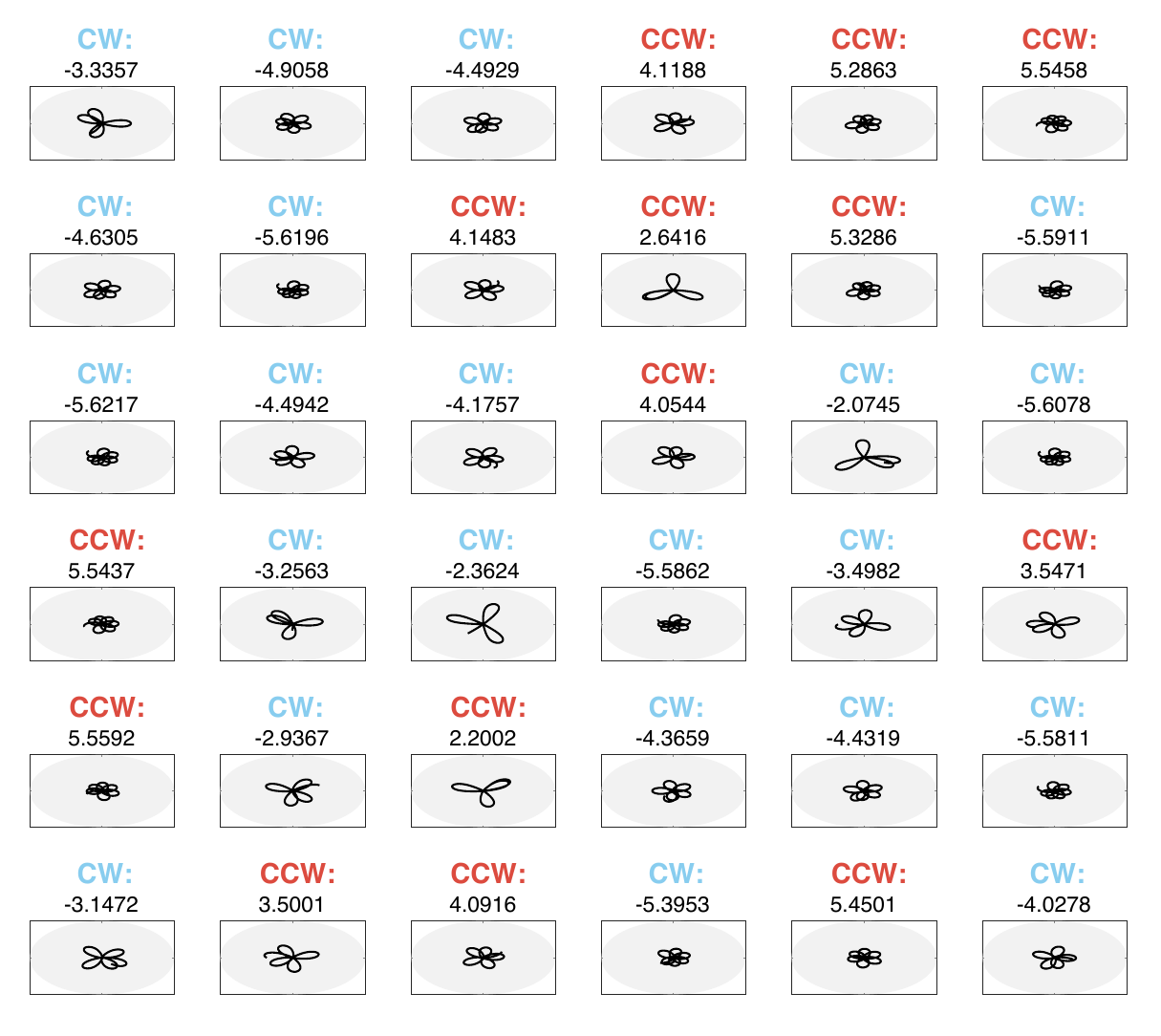}
    \caption{Sample trajectories for simulated unbiased ($\mu=0$) singlets.}
    \label{}
\end{figure*}

\begin{figure*}[h!]
    \centering
    \includegraphics[width=\textwidth]{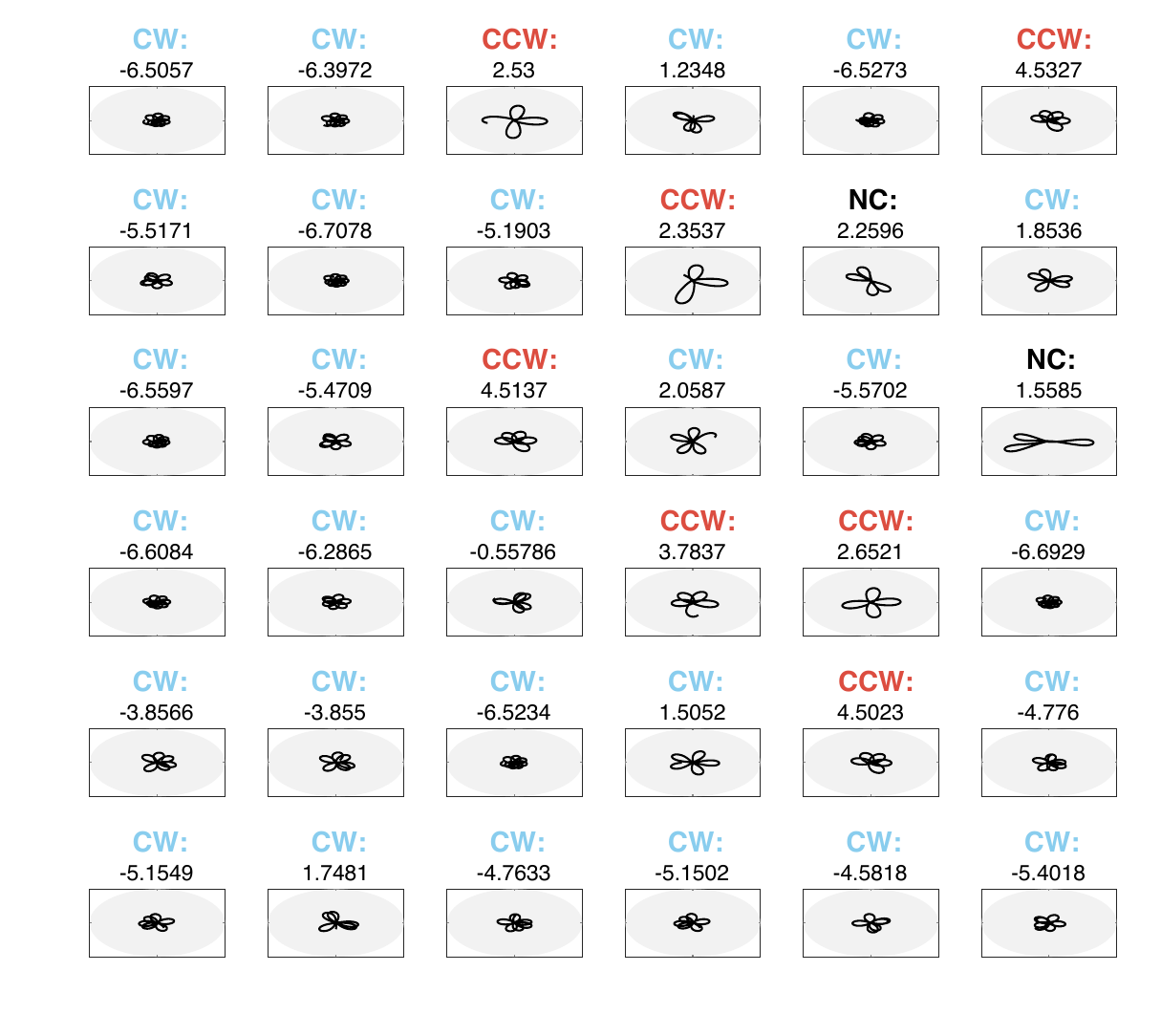}
    \caption{Same as Fig.~S2 but for biased ($\mu=-0.5$) singlets.}
    \label{}
\end{figure*}

\begin{figure*}[h!]
    \centering
    \includegraphics[width=\textwidth]{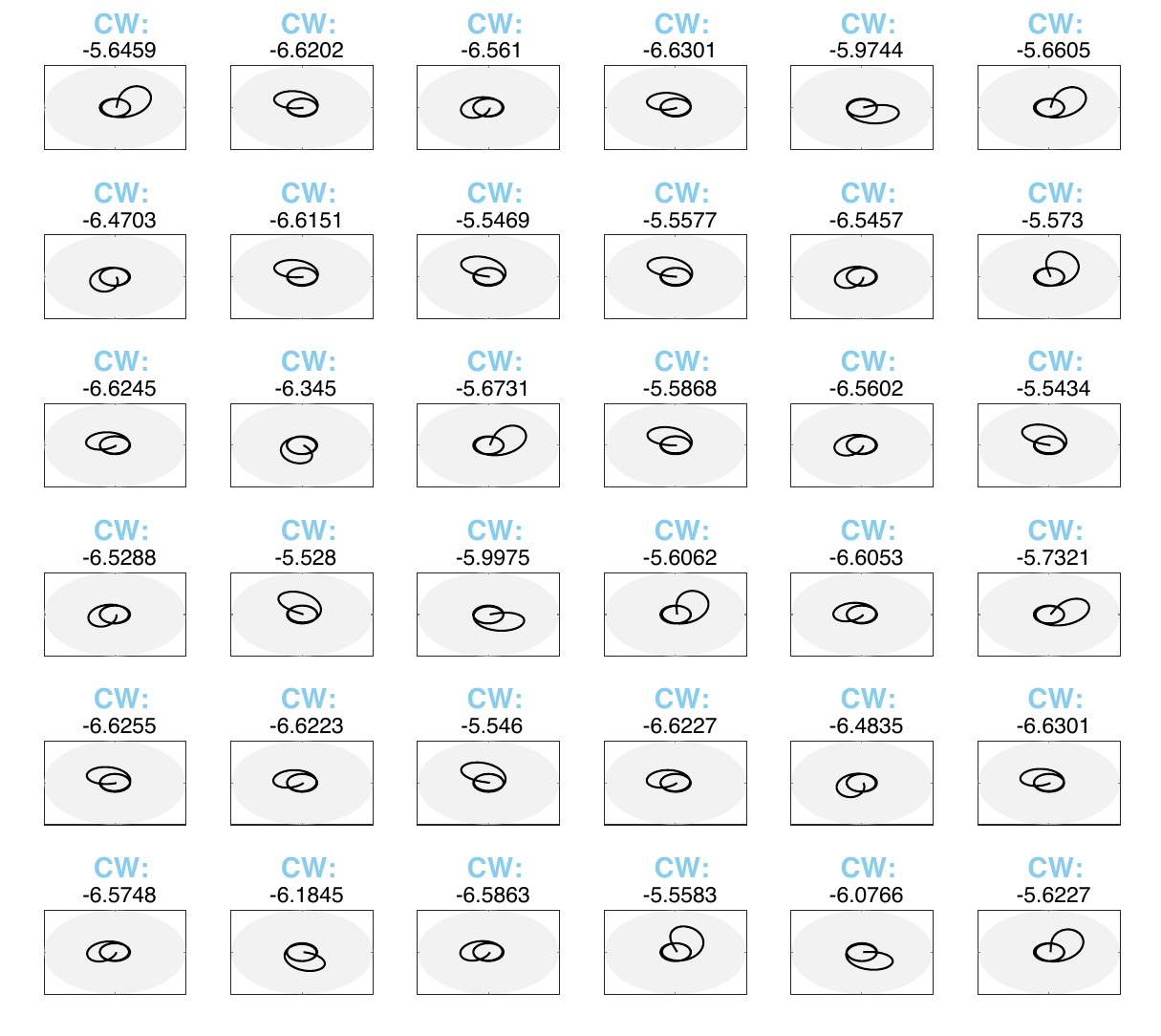}
    \caption{Same as Fig.~S2 (unbiased singlets), but with an easy frictional direction at an offset angle $\delta=\pi/4$ from the polarity axis.}
    \label{}
\end{figure*}

\begin{figure*}[h!]
    \centering
    \includegraphics[width=0.95\linewidth]{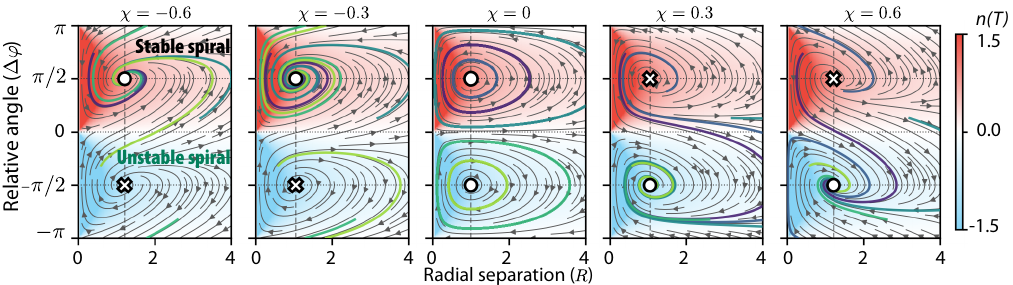}
    \caption{Phase portraits of the singlet model with chiral cell–wall interactions modulated by parameter $\chi$. Depending on the sign of $\chi$, one rotational equilibrium loses stability. White circles denote stable (or neutrally stable) equilibria, white crosses unstable equilibria, the background colormap indicates the time-averaged winding number $n(T)$, and solid colored curves show representative trajectories.}
    \label{fig:chiralWall}
\end{figure*}

\begin{figure}[ht]
    \centering
    \includegraphics[width=\textwidth]{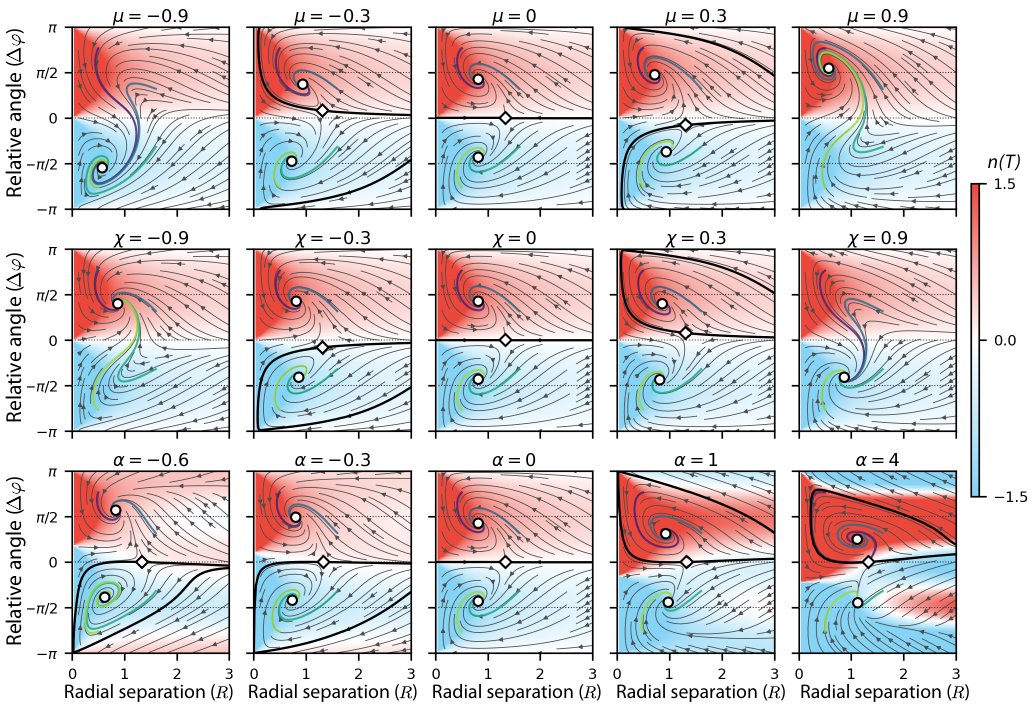}
    \caption{Phase diagrams of the anchored singlet combined with each symmetry-breaking coupling: intrinsic bias $\mu$ (top row), chiral wall $\chi$ (middle row), and anisotropic mobility $(\alpha)$ (bottom row). Parameters are $k_W = 1.0, k_0 = 0.5$ so that $\kappa = 1/2$ and $R_{\rm disc} = 1$. White circles denote stable (or neutrally stable) equilibria, the background colormap indicates the time-averaged winding number $n(T)$, and solid colored curves show representative trajectories.}
    \label{fig:anchor}
\end{figure}

\begin{sidewaysfigure*}[h!]
    \centering
    \includegraphics[width=\textwidth]{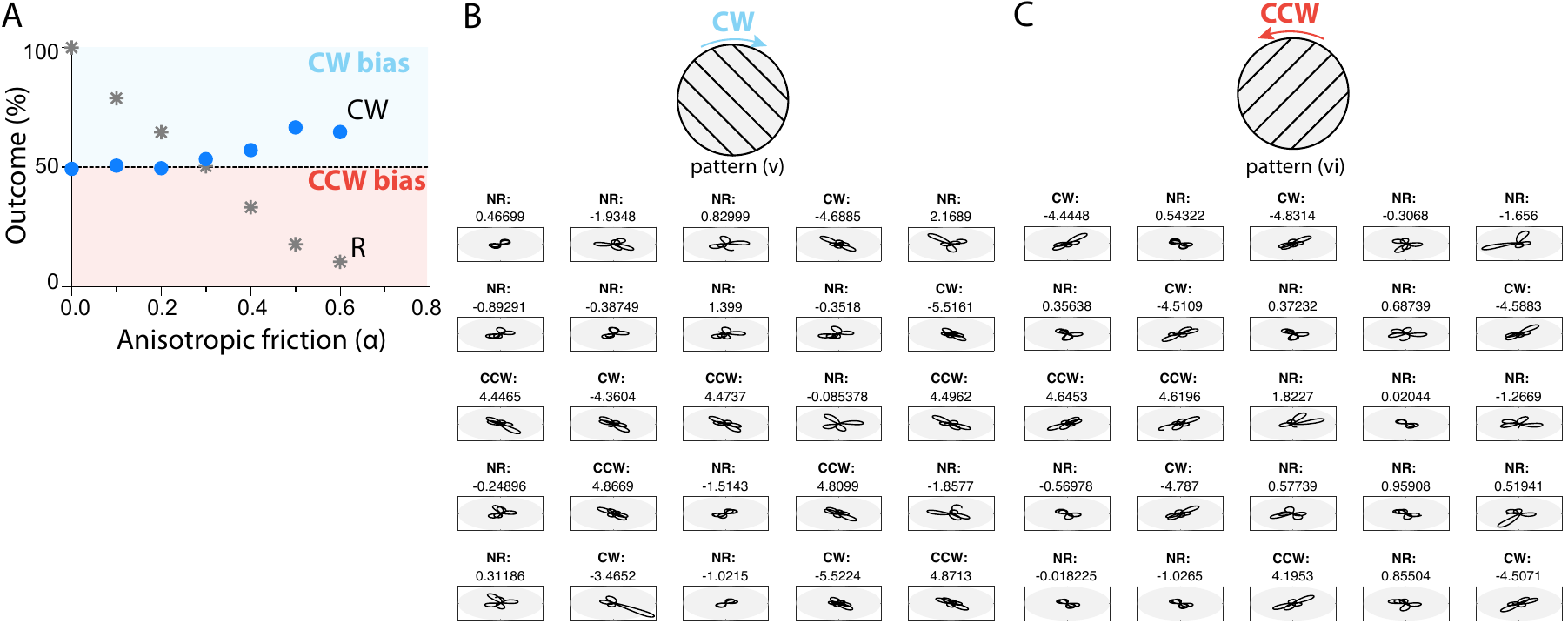}
    \caption{(A) Motility outcomes shown as percentages for varying anisotropic frictional coefficient ($\alpha$) for unbiased cells migrating on pattern (v). Grey stars mark percentage of rotating cells, while blue dots indicate clockwise percentage. (B-C) Sample trajectories for simulated unbiased singlets on (B) pattern (v) and (C) pattern (vi).}
    \label{}
\end{sidewaysfigure*}

\clearpage

\appendix
\renewcommand{\thesection}{Appendix~\Alph{section}}
\renewcommand{\thesubsection}{\Alph{section}.\arabic{subsection}}
\renewcommand{\thetheorem}{\Alph{section}.\arabic{theorem}}

\input{appendix}

%% file: appendix.tex
\section{The general singlet model}
\label{sec:general}

This appendix collects the dynamical systems analysis of the singlet model.
We first state the full model with every mechanism switched on and reduce it to
a closed autonomous flow on the phase lag $(R, \Dphi)$. Each subsequent
section is then a specialization obtained by switching off a subset of parameters, allowing us to focus on one mechanism at a time.

\subsection{Equations of motion}

We consider a single self-propelled point particle confined to a
two-dimensional disk of radius $R_{\rm disc}$. Its lab-frame state is the
position $\vb{x}(t) \in \mathbb{R}^2$, written in polar coordinates as
$\vb{x} = R\lb\cos\theta, \sin\theta\rb^T$, along with a body polarity
angle $\varphi(t)$ and a unit vector $\phat = \lb\cos\varphi, \sin\varphi\rb^T$.

The active propulsion of magnitude $\gamma_\text{pol}$ along $\phat$
drives the cell, and the polarity's inward reorientation at the boundary (the wall-alignment torque, below) is what confines it to the disk; even with no radial force the cell settles onto a bounded orbit (Appendix~B). We allow, in addition, an optional radial force $f(R)\rhat$,
\begin{equation}\label{eq:gen_fR}
    f(R) = -k_0 R + k_W \lb R_{\rm disc} - R \rb,
\end{equation}
a wall spring $k_W \geq 0$ that restores the cell toward the boundary $R_{\rm disc}$ together with a central anchor $k_0$ ($k_0 > 0$ attracting, $k_0 < 0$ repelling). It is switched on only in Appendix~F, where this force cooperates with the polarization reorientation; otherwise $f \equiv 0$. The total mechanical force is 
\begin{equation}\label{eq:gen_F}
    \vb{F} = \gamma_\text{pol}\phat + f(R)\rhat.
\end{equation}
The equations of motion are overdamped force balance for the position and a
three-torque balance for the polarity,
\begin{align}
    \dot{\vb{x}} &= \vb{v} = \mathbb{M}\vb{F}, \label{eq:gen_xdot}\\
    \dot\varphi &= \frac{\abs{\vb{v}}\sin\lb\theta_V - \varphi\rb}{\tau_{\rm VA} V_0}
        + \frac{\sin\lb\theta_W - \varphi + \chi\rb}{\tau_W} + \mu, \label{eq:gen_phidot}
\end{align}
where $\mathbb{M}$ is the mobility tensor (defined below), $\theta_V = \arg\vb{v}$
is the velocity direction, $\theta_W = \theta + \pi$ is the inward radial direction, and 
$V_0 = \gamma_\text{pol}/\xi_\perp$ is the free-moving speed. The polarity
feels three torques: velocity alignment toward the heading $\theta_V$
(with timescale $\tau_{\rm VA}$), wall alignment toward $\theta_W + \chi$
(with timescale $\tau_W$, and an optional body-frame chiral offset $\chi$), and an optional intrinsic torque $\mu$. Notably, we write the velocity-alignment term in the
\emph{signed} form
\begin{equation}\label{eq:gen_perp}
    \abs{\vb{v}}\sin\lb\theta_V - \varphi\rb
        = \lb \phat \times \vb{v}\rb\cdot\vb{z},
\end{equation}
the $\vb{z}$-component of the planar cross product. Unlike the
normalized torque $\sin\lb\theta_V - \varphi\rb$, this is smooth at stagnation
$\vb{v} \to \vb{0}$: a stationary cell feels no alignment torque.

\subsection{Mobility tensor}

We allow the cell's friction with the substrate to be anisotropic as described in Eq.~(8): sliding along an easy axis $\ehat$ (with friction coefficient $\xi_\parallel$) is easier than across it (with friction coefficient $\xi_{\perp}$). 
The easy axis is either fixed in the lab at a constant angle $\beta$,
$\ehat = \lb\cos\beta, \sin\beta\rb^T$, or co-rotating with the polarity at a fixed body-frame offset $\delta$, $\ehat = \lb\cos\lb\varphi+\delta\rb, \sin\lb\varphi+\delta\rb\rb^T$.

\subsection{Reduction to $(R, \Dphi)$}

We define the phase lag
\begin{equation}\label{eq:gen_Dphi}
    \Dphi \equiv \varphi - \theta,
\end{equation}
the angular offset between polarity and position angle. At $(R, \theta)$ the local orthonormal basis is $\rhat = \lb\cos\theta, \sin\theta\rb^T$, $\thetahat = \lb-\sin\theta, \cos\theta\rb^T$, and the polarity has components $\phat^T\rhat = \cos\Dphi$, $\phat^T\thetahat = \sin\Dphi$. The total force in Eq.~\eqref{eq:gen_F} has local components
\begin{equation}\label{eq:gen_Fcomp}
    F_r = \gamma_\text{pol}\cos\Dphi + f(R), 
    \quad
    F_\theta = \gamma_\text{pol}\sin\Dphi,
\end{equation}
and the wall torque simplifies to
$\sin\lb\theta_W - \varphi + \chi\rb = \sin\lb\pi - \Dphi + \chi\rb = \sin\lb\Dphi - \chi\rb$,
the chiral offset surviving the reduction; $\chi = 0$ reduces to the  wall torque $\sin\Dphi$.

Let $\psi$ be the angle between $\rhat$ and $\ehat$, so $\ehat^T\rhat = \cos\psi$
and $\ehat^T\thetahat = \sin\psi$; the co-rotating case has $\psi = \Dphi + \delta$
(a function of the reduced state) while the lab-frame case has
$\psi = \beta - \theta$. Applying $\mathbb{M}$ to $\vb{F}$ and projecting onto the local basis,
\begin{align}
    \xi_\perp v_r &= F_r + \alpha \lb\vb{F}^T\ehat\rb\cos\psi, \label{eq:gen_vr}\\
    \xi_\perp v_\theta &= F_\theta + \alpha \lb\vb{F}^T\ehat\rb\sin\psi, \label{eq:gen_vtheta}
\end{align}
where $v_r = \dot R$, $v_\theta = R\dot\theta$, and
$\vb{F}^T\ehat = \gamma_\text{pol}\cos\lb\psi - \Dphi\rb + f(R)\cos\psi$.

\begin{lemma}[Closure of the reduced system]\label{lem:gen_closure}
Eqs.~\eqref{eq:gen_xdot}--\eqref{eq:gen_phidot}, written on the state
$(R, \Dphi)$, define a closed autonomous flow if and only if
$\psi = \psi(R, \Dphi)$ is a function of the reduced state alone. The
co-rotating case ($\psi = \Dphi + \delta$) closes; the lab-frame case
($\psi = \beta - \theta$) does not, and exists on the full $(R, \Dphi, \theta)$ manifold.
\end{lemma}

For the co-rotating substrate ($\psi = \Dphi + \delta$, or
$\vb{F}^T\ehat = \gamma_\text{pol}\cos\delta + f(R)\cos\psi$), the velocity
components are
\begin{equation}\label{eq:gen_master_v}
\begin{aligned}
    \xi_\perp v_r &= \gamma_\text{pol}\lb \cos\Dphi + \alpha\cos\delta\cos\psi\rb
                     + f(R)\lb 1 + \alpha\cos^2\psi\rb,\\
    \xi_\perp v_\theta &= \gamma_\text{pol}\lb \sin\Dphi + \alpha\cos\delta\sin\psi\rb
                     + f(R)\alpha\cos\psi\sin\psi,
\end{aligned}
\end{equation}
and, using the signed alignment term in Eq.~\eqref{eq:gen_perp}, the closed \emph{fundamental reduced system} on $(R, \Dphi)$ becomes
\begin{equation}\label{eq:gen_master}
\boxed{
    \begin{aligned}
        \dot R &= v_r,\\
        \dot{\Dphi} &= \frac{\sin\lb\Dphi - \chi\rb}{\tau_W}
            + \frac{v_\theta\cos\Dphi - v_r\sin\Dphi}{\tau_{\rm VA} V_0}
            - \frac{v_\theta}{R} + \mu,
    \end{aligned}}
\end{equation}
where $V_0 = \gamma_\text{pol}/\xi_\perp$ is the free-moving speed.\\

Every section below is a specialization of Eq.~\eqref{eq:gen_master}, each obtained by switching a subset of these five mechanisms on or off: the intrinsic torque $\mu$, the cell-substrate frictional anisotropy $\alpha$, the body-frame offset $\delta$, the chiral confinement offset $\chi$, and the presence of a radial force $f(R)$ (centering anchor $k_0$ and wall $k_W$). A summary of the stability analysis of steady states in each system is provided in Table~S2.

\subsection{Reflection symmetry}

\begin{lemma}[Reflection symmetry of the master reduced system]\label{lem:Z2}
Consider the fundamental reduced system in Eq.~\eqref{eq:gen_master} under the reflection $(R, \Dphi) \mapsto (R, -\Dphi)$. If the chiral couplings vanish ($\mu = 0$, $\chi = 0$, and $\delta = 0$, i.e.\ the easy axis aligned with the polarity), then for arbitrary anisotropy $\alpha > -1$ and radial force $f(R)$ the vector field obeys $(\dot R, \dot\Dphi) \mapsto (\dot R, -\dot\Dphi)$: the flow is invariant and
every trajectory $(R(t), \Dphi(t))$ has a mirror image $(R(t), -\Dphi(t))$. A nonzero $\mu$, $\chi$, or $\delta$ (the last with $\alpha \neq 0$) breaks it.
\end{lemma}

\begin{proof}
At $\delta = 0$ we have $\psi = \Dphi$, so by Eq.~\eqref{eq:gen_master_v} every term of $v_r$ carries $\cos\Dphi$ or $\cos^2\Dphi$ (even) and every term of $v_\theta$
carries $\sin\Dphi$ or $\cos\Dphi\sin\Dphi$ (odd), for any $\alpha$ and $f(R)$. In $\dot\Dphi$ the wall term $\sin(\Dphi-\chi)/\tau_W$ is odd iff $\chi = 0$; the combination $v_\theta\cos\Dphi - v_r\sin\Dphi$ and the term $-v_\theta/R$ are odd; and the constant $\mu$ is the sole even contribution. Hence $\dot R$ is even and $\dot\Dphi$ is odd precisely when $\mu = \chi = \delta = 0$, which gives the symmetry; any nonzero chiral coupling introduces an even term into $\dot\Dphi$ and breaks it.
\end{proof}

\subsection{Phase space and convergence criterion}

Every reduced flow below exists on the same phase space, and the same three conditions yield global stability in each case. We verify them once here; each specialization then invokes the conclusion instead of repeating the argument. The reduced state $(R, \Dphi)$ is an element in the half-cylinder $\mathcal{M} = \mathbb{R}_{>0}\times S^1$, with $\Dphi \in (-\pi, \pi]$. The disk center is the single orbit $R \to 0$; adjoining it as one point $*$, the map
\begin{equation}\label{eq:gen_Phi}
    \Phi(R, \Dphi) = \lb R\cos\Dphi,\, R\sin\Dphi\rb, \quad \Phi(*) = (0, 0),
\end{equation}
is a homeomorphism $\mathcal{M}\cup\lcb*\rcb \cong \mathbb{R}^2$. Topological notions on $\mathcal{M}\cup\lcb*\rcb$ are therefore inherited from $\mathbb{R}^2$; in particular Heine--Borel holds, so a closed disc $D_{R_{\max}} = \lcb*\rcb \cup \lcb (R, \Dphi) : R \leq R_{\max}\rcb$ is compact, and the planar Poincar\'e--Bendixson and Bendixson--Dulac theorems apply.

\begin{lemma}[Convergence on an invariant region]\label{lem:toolkit}
Let $\Omega \subseteq \mathcal{M}$ be open, simply connected and forward-invariant,
with every forward trajectory bounded. Suppose
\begin{enumerate}
    \item some $\eta \in C^1(\Omega)$ s.t.\ $\nabla\cdot\lb\eta\,\vb{F}\rb > (<) 0$
    on $\Omega$; and
    \item $\Omega$ contains a single equilibrium, hyperbolic and not a saddle.
\end{enumerate}
Then that equilibrium is a sink, and every trajectory in $\Omega$ converges to it.
\end{lemma}

\begin{proof}
Each forward orbit is bounded, so its $\omega$-limit set is a nonempty compact subset of $\mathcal{M}$. Hypothesis~(1) excludes periodic orbits (Bendixson--Dulac), and with a single non-saddle equilibrium there are no homoclinic or heteroclinic connections. By Poincar\'e--Bendixson every $\omega$-limit set is the lone equilibrium, which is therefore a sink attracting all of $\Omega$.
\end{proof}

\noindent Most of the symmetry-breaking sections below reduce to checking these hypotheses: a coupling that funnels the flow into a single half-cylinder leaves one global attractor (Appendix~E), whereas one that keeps the axes $\Dphi \in \lcb 0, \pi \rcb$ invariant splits $\mathcal{M}$ into two strips, each with its own sink, and the system is bistable (Appendix~F). The anisotropic friction of Appendix~D is the exception.

\section{Dynamical system analysis of the unbiased singlet}
\label{sec:symmetric}
We start with the simplest case: all five mechanisms of Appendix~A are switched off. The particle is unbiased ($\mu = 0$) on an isotropic substrate ($\alpha = 0$ yielding $\mathbb{M} = I/\xi$ with a single friction coefficient $\xi \equiv \xi_\perp$), with no additional cell-substrate forces ($f(R) \equiv 0$) and an achiral wall interaction ($\chi = 0$); confinement comes from the wall alignment torque alone. The particle position and lab-frame state $\lb\vb{x}, \varphi\rb$ are as defined in Appendix~A, Eqs.~\eqref{eq:gen_xdot}--\eqref{eq:gen_phidot}.

\subsection{Reduced system $(R, \Dphi)$}
Setting $\alpha = f = \chi = \mu = 0$ in Eq.~\eqref{eq:gen_master} makes the velocity parallel to the polarity,
\begin{align}
    v_r &= V_0\cos\Dphi, \label{eq:vr_raw}\\
    v_\theta &= V_0\sin\Dphi, \label{eq:vtheta_raw}
\end{align}
where $V_0 = \gamma_\text{pol}/\xi$ is the free-moving speed, $v_r = \dot R$, and $v_\theta = R\dot\theta$. 
The velocity alignment torque drops out, and by writing $\ell = V_0\tau_W$ for 
the distance covered in one wall alignment time, the closed flow on $(R, \Dphi)$ becomes
\begin{equation}\label{eq:sym_reduced}
    \dot R = V_0\cos\Dphi,
    \quad
    \dot{\Dphi} = \sin\Dphi\lb \frac{1}{\tau_W} - \frac{V_0}{R}\rb.
\end{equation}


\subsection{Conserved Hamiltonian quantity}

\begin{proposition}[Hamiltonian]\label{prop:H}
The system in Eq.~\eqref{eq:sym_reduced} admits the first integral
$H(R, \Dphi) = R\,e^{-R/\ell}\sin\Dphi$.
\end{proposition}

\begin{proof}
The partial derivatives of this Hamiltonian quantity are
\begin{equation}
    \frac{\partial H}{\partial R} = e^{-R/\ell}\lb 1 - \frac{R}{\ell}\rb\sin\Dphi,
\end{equation}
\begin{equation}
    \frac{\partial H}{\partial \Dphi} = R\,e^{-R/\ell}\cos\Dphi.
\end{equation}
Substituting into Eq.~\eqref{eq:sym_reduced} and using $V_0/\ell = 1/\tau_W$, yields
\begin{align*}
    \dot H &= e^{-R/\ell}\lb 1 - \tfrac{R}{\ell}\rb\sin\Dphi \cdot V_0\cos\Dphi \\
    &\quad + R\,e^{-R/\ell}\cos\Dphi \cdot \sin\Dphi\lb \tfrac{1}{\tau_W} - \tfrac{V_0}{R}\rb \\
    &= e^{-R/\ell}\sin\Dphi\cos\Dphi \cdot R\lb \tfrac{1}{\tau_W} - \tfrac{V_0}{\ell}\rb = 0. \qedhere
\end{align*}
\end{proof}
\noindent Trajectories therefore lie on level sets of $H$, which makes the phase portrait (Fig.~2B,E,G) of Eq.~\eqref{eq:sym_reduced} fully transparent. Two properties of $H$ deserve highlighting. First, the prefactor $R\,e^{-R/\ell}$ is strictly positive for $R > 0$, so the sign of $H$ equals the sign of $\sin\Dphi$. The lines $\Dphi = 0$ and
$\Dphi = \pm\pi$ are therefore invariant (level set $H = 0$), and they separate phase space into an upper half-plane $H > 0$ and a lower half-plane $H < 0$. Second, the level sets are bounded in $R$: as $R \to \infty$ the prefactor decays to zero, so each $H \neq 0$ level curve has a maximum $R$ on its orbit.

\subsection{Equilibria}

\begin{proposition}[Equilibria of the unbiased system]\label{prop:sym_eq}
The reduced system in Eq.~\eqref{eq:sym_reduced} has:
\begin{enumerate}
    \item Two interior centers at $(R^*, \Dphi^*) = (\ell, \pm\pi/2)$.
    Each corresponds to steady circular motion of radius $\ell$ at
    angular velocity $\dot\theta = \pm 1/\tau_W$.

    \item Two invariant lines (not equilibria) at $\Dphi = 0$ and $\Dphi = \pm\pi$, 
    forming the $H = 0$ level set, where $\dot\Dphi = 0$ but $\dot R = \pm V_0 \neq 0$.
    The cell drifts radially outward ($\Dphi = 0$) or inward ($\Dphi = \pm\pi$). They are
    not equilibria, but they organize the phase portrait.
    
\end{enumerate}
\end{proposition}

\begin{proof}
From $\dot R = 0$ either $\cos\Dphi = 0$ (giving $\Dphi = \pm\pi/2$),
or the trajectory has unbounded $R$. In the first case,
$\dot{\Dphi} = 0$ requires $\sin\Dphi = 0$ or $R = \ell$; the
former contradicts $\Dphi = \pm\pi/2$, leaving $R^* = \ell$. The
Jacobian at $(\ell, \pm\pi/2)$ is
\begin{equation}
    J = \begin{pmatrix} 0 & \mp V_0 \\ \pm \frac{1}{\ell\tau_W} & 0 \end{pmatrix},
\end{equation}
giving $\lambda^2 = -1/\tau_W^2$ and frequency $\omega = 1/\tau_W$, yielding a center. At $\Dphi = \pm\pi/2$, Eq.~\eqref{eq:vtheta_raw} can be re-written as $\dot\theta = \pm V_0/R^* = \pm 1/\tau_W$. The lines $\Dphi = 0$ and $\Dphi = \pm\pi$ are invariant under Eq.~\eqref{eq:sym_reduced}, since $\sin\Dphi = 0$ means $\dot{\Dphi} = 0$. On these lines $\dot R = \pm V_0$, meaning that trajectories travel radially outward ($\Dphi = 0$) or inward toward the origin ($\Dphi = \pm\pi$). They are not orbits of finite-energy centers, but they organize the phase portrait: every level set $H = c$ with small $|c|$ closely follows them at large $R$.
\end{proof}


\subsection{Classification of trajectories}

The conserved quantity $H$ classifies trajectories. Using
$R\dot\theta = V_0\sin\Dphi$ from Eq.~\eqref{eq:vtheta_raw}, the
angular velocity of the particle along any trajectory is
\begin{equation}\label{eq:omega_H}
    \omega = \dot\theta = \frac{V_0\sin\Dphi}{R}
    = \frac{V_0 e^{R/\ell}}{R^2}\,H,
\end{equation}
so the sign of $\omega$ is the sign of $H$ for all $R > 0$.

\begin{corollary}[Rotation direction classifier]\label{cor:H_classifier}
Along any trajectory of Eq.~\eqref{eq:sym_reduced} with initial condition
$(R_0, \Dphi_0)$:
\begin{enumerate}
    \item If $H(R_0, \Dphi_0) > 0$, then $\omega(t) > 0$ for all $t$,
    and $n(T) > 0$ for all $T > 0$.
    \item If $H(R_0, \Dphi_0) < 0$, then $\omega(t) < 0$ for all $t$,
    and $n(T) < 0$ for all $T > 0$.
    \item If $H(R_0, \Dphi_0) = 0$, then $\omega = 0$: the
    trajectory lies on $\Dphi \in \{0, \pm\pi\}$ and $n(T) = 0$.
\end{enumerate}
No trajectory mix CW and CCW oriented migration paths.
\end{corollary}

\begin{proof}
$H$ is conserved along trajectories, and Eq.~\eqref{eq:omega_H} gives
$\mathrm{sgn}(\omega) = \mathrm{sgn}(H)$. The sign of $\omega(t)$ is
therefore constant, so $\Theta(T) = \int_0^T \omega\,dt$ inherits the
same sign, and hence so does $n(T) = \Theta(T)/(2\pi)$.
\end{proof}

\section{Effect of intrinsic torque}\label{sec:bias}

The simplest $\mathbb{Z}_2$-breaking term is a constant intrinsic
torque $\mu \neq 0$ in the polarity equation. It enters $\dot\varphi$ as a constant, so it adds $\mu$ to $\dot{\Dphi}$ in the reduced equations
\begin{equation}\label{eq:bias_reduced}
    \dot R = \frac{\gamma_\text{pol}}{\xi}\cos\Dphi,
    \quad
    \dot{\Dphi} = \sin\Dphi\lb \frac{1}{\tau_W} - \frac{V_0}{R}\rb + \mu.
\end{equation}

\begin{proposition}[Equilibria of the biased system]\label{prop:bias_eq}
For $|\mu| < 1$, the system above has two interior
equilibria:
\begin{equation}
\begin{gathered}\label{eq:int_eq_mu}
(R^*_+, \Dphi^*_+) = \lb \frac{V_0}{1/\tau_W + \mu}, +\frac{\pi}{2}\rb,\\
(R^*_-, \Dphi^*_-) = \lb \frac{V_0}{1/\tau_W - \mu}, -\frac{\pi}{2}\rb.
\end{gathered}
\end{equation}
The first exists for $\mu > -1/\tau_W$ and the second for $\mu < +1/\tau_W$. Both are centers, with linearization frequencies $\omega_+ = |1/\tau_W + \mu|$ and $\omega_- = |1/\tau_W - \mu|$.
\end{proposition}

\begin{proof}
Setting $\dot R = 0$ requires $\cos\Dphi = 0$, hence $\Dphi = \pm\pi/2$. At $\Dphi = +\pi/2$, $\dot{\Dphi} = 0$ becomes
$1/\tau_W - V_0 R^* + \mu = 0$, giving $R^*_+ = V_0/(1/\tau_W+\mu)$. At $\Dphi = -\pi/2$, the equation is $-(1/\tau_W - V_0 R^*) + \mu = 0$, giving $R^*_- = V_0/(1/\tau_W-\mu)$. The Jacobian at
$(R^*_+, \pi/2)$ is
\begin{equation}
    J_+ = \begin{pmatrix} 0 & -V_0 \\ V_0/(R^*_+)^2 & 0 \end{pmatrix}
        = \begin{pmatrix} 0 & -V_0 \\ (1/\tau_W+\mu)^2/V_0 & 0 \end{pmatrix},
\end{equation}
yielding $\lambda^2 = -(1/\tau_W+\mu)^2$ and $\omega_+ = |1/\tau_W+\mu|$. The calculation for $(R^*_-, -\pi/2)$ is analogous.
\end{proof}

The presence of the intrinsic torque also shifts the lines $\Dphi \in \{0, \pm\pi\}$. The lines are no longer invariant: with $\sin\Dphi = 0$, the $\Dphi$ equation gives $\dot{\Dphi} = \mu \neq 0$, so trajectories cross these lines at constant rate $\mu$. The partition of phase space into upper and lower half-planes is therefore broken.

\subsection{Global bifurcation structure}

%
%
\begin{theorem}[Global bifurcation]\label{thm:bif}
As $\mu$ varies, the centers in Eq.~\eqref{eq:int_eq_mu} undergo:
\begin{enumerate}
    \item At $\mu \to -1/\tau_W^+$: the CCW center radius $R^*_+ = V_0/(1/\tau_W+\mu)$ diverges, and the equilibrium escapes to infinity. For $\mu \leq -1/\tau_W$ no CCW equilibrium remains.
    
    \item At $\mu = 0$: the symmetric configuration of Appendix B is recovered, with mirror-paired centers at $(\tau_W V_0, \pm\pi/2)$.
    
    \item At $\mu \to +1/\tau_W^-$: the CW center radius $R^*_- = V_0/(1/\tau_W-\mu)$ diverges, and the equilibrium escapes to infinity. For $\mu \geq +1/\tau_W$ no CW equilibrium remains.
\end{enumerate}
For $|\mu| > 1/\tau_W$ only one rotation direction survives in the
finite-$R$ phase plane. The dynamics is deterministically chiral.
\end{theorem}

\begin{proof}
The radii in Eq.~\eqref{eq:int_eq_mu} diverge precisely at $\mu = -1/\tau_W$ and $\mu = +1/\tau_W$. The frequencies $\omega_\pm$ vanish there, since the centers spiral outward arbitrarily slowly as they escape to infinity.
\end{proof}

\noindent In this case we can define a new Hamiltonian quantity: 
$\widetilde{H}(R, \Dphi) := e^{-R / V_0 \tau_W}\lsb R \sin\Dphi + \mu\lb R + \gamma_\text{pol} \rb\rsb.$
A similar analysis as in Appendix B can be followed to show that $\dot{\widetilde{H}} = 0$ along solution trajectories.\\

\noindent Next, we want to compute the sign of the instantaneous angular velocity $\omega$.
On a level set $\widetilde{H} = C$ we have
\begin{equation}
    \sin\Dphi = \frac{e^{R/\gamma_\text{pol}}}{R}\lsb C + \mu e^{-R/\gamma_\text{pol}}\lb R+ \gamma_\text{pol} \rb \rsb
\end{equation}
Define $g(R) = e^{-R/\gamma_\text{pol}}\lb R+ \gamma_\text{pol} \rb$ then $\sin\Dphi = \frac{e^{R/\gamma_\text{pol}}}{R}\lsb C - g(R) \rsb$. Next, we observe that $\max g = \mu \gamma_\text{pol}$ and $\min g = 0$. Three distinguishable regimes are defined:
\begin{itemize}
    \item If $C > \mu \gamma_\text{pol}$, $\sgn\sin\Dphi > 0$ and $\omega > 0$;
    \item If $0 < C < \mu\gamma_\text{pol}$, $\omega$ takes both sign along a trajectory; 
    \item If $C < 0$ $\sgn\sin\Dphi < 0$, and $\omega < 0$. 
\end{itemize}
In the intermediate band ($0<C<\mu\ell$), the sign of $\dot\Dphi$ is fixed, so $\Dphi$ advances through $2\pi$ in one period and the net
rotation is its time-averaged angular velocity,
\begin{equation}\label{eq:bias_omega}
    \langle\omega\rangle = \frac{\Theta(C)}{T(C)},
    \quad
    T(C) = \oint \frac{d\Dphi}{\dot\Dphi},
    \quad
    \Theta(C) = \oint \omega\,\frac{d\Dphi}{\dot\Dphi},
\end{equation}
with $\omega = \ell\sin\Dphi/R$, $\dot\Dphi = \sin\Dphi\lb 1-\ell/R\rb + \mu$, and
$R=R(\Dphi;C)$ from inverting $\widetilde H = C$; the winding grows linearly, $n(T)\to (T/2\pi)\langle\omega\rangle$. We evaluate these integrals by quadrature. In this band the cell can rotate against its own bias: for $0<C<\mu\ell$ with $\mu>0$ one finds $\langle\omega\rangle<0$ (net clockwise though $\mu>0$ favours counter-clockwise), which reshapes the basins in Fig.~2G.

\section{Generalization to any mobility matrix}
\label{sec:aniso}

We now turn on the anisotropic mobility, an easy axis frictional anisotropy $\alpha \neq 0$ carried at a fixed co-rotating offset $\delta \neq 0$, while leaving the remaining couplings off ($f = \chi = \mu = 0$). This is the co-rotating anisotropic specialization of the fundamental reduced system in Eq.~\eqref{eq:gen_master} (mobility tensor Eq.~(8), closure $\psi = \Dphi + \delta$ of Lemma Eq.~\ref{lem:gen_closure}). The offset $\delta \neq 0$ breaks the reflection symmetry of Appendix~B and destroys the conserved quantity: the flow becomes dissipative, with two hyperbolic equilibria, one stable and one unstable, and hence a single global attractor.\\

\noindent Altering the fundamental system in Eq.~\eqref{eq:gen_master} to this case ($f = \chi = \mu = 0$, $\psi = \Dphi + \delta$) gives the reduced $(R, \Dphi)$ system
\begin{equation}\label{eq:aniso_singlet}
\begin{aligned}
    \dot{R} &= \ell\lb A \cos\Dphi - B \sin\Dphi\rb,\\
    \dot{\Dphi} &= \sin\Dphi + B - \frac{\ell}{R}\lb A \sin\Dphi + B \cos\Dphi\rb,
\end{aligned}
\end{equation}
on the cylinder $\mathcal{M} = \mathbb{R}_{>0}\times S^1$, where $A = 1 + \alpha\cos^2\delta$, $B = \frac{\alpha}{2}\sin\lb2 \delta\rb$, $\rho = \sqrt{A^2 + B^2}$, and $\ell = V_0 \tau_W$.
We take $\alpha >0$ (so $B > 0$).
The constant $B$ in $\dot{\Dphi}$ is the velocity-alignment torque of the Eq.~(8). For the co-rotating anisotropy it collapses to the constant $B = (\alpha/2)\sin\lb 2 \delta \rb$, so the same mobility coefficient that tilts the velocity reappears as an emergent bias on the polarity. Here time 
is measured in units of $\tau_W$ and we set $\tau_{\rm VA} = \tau_W$, so both alignment rates
are unity. 

\subsection{Equilibria}
Write $\vb{F} = (P, Q)$ for the vector field of Eq.~\eqref{eq:aniso_singlet} on the phase space 
$\mathcal{M}$ and its compactification $\mathcal{M}\cup\lcb*\rcb$ of Appendix~A, 
and split $\mathcal{M}$ into the open half-cylinders $U = \mathbb{R}_{>0}\times(0,\pi)$ 
and $L = \mathbb{R}_{>0}\times(-\pi,0)$, separated by the axes 
$\Gamma_0, \Gamma_\pi$ at $\Dphi = 0, \pi$.

\begin{theorem}\label{thm:aniso}
    For $\alpha > 0$ and $\delta \in (0, \pi/2)$, the anisotropic singlet in Eq.~\eqref{eq:aniso_singlet} has at most two equilibria
    $\vb{x}_{\pm} = (R_{\pm},\Dphi_{\pm})$. 
    They lie where the velocity is purely tangential (perpendicular to the radial 
    direction $\mathbf{\hat{r}}$), with angles and radii
    \begin{equation}
        \Dphi_{\pm} = \mp\frac{\pi}{2} - \phi_0,
        \quad 
        \phi_0 = \arctantwo\lb B, A \rb \in \lb 0, \frac{\pi}{2} \rb, 
        \quad 
        R_{\pm} = \frac{\ell\rho^2}{A\mp B \rho},
    \end{equation}
    and Jacobian invariants
    \begin{equation}
        \det J_{\pm} = \frac{\lb A \mp B \rho\rb^2}{\rho^2} > 0,
        \quad
        \tr J_{\pm} = \mp\frac{\alpha\sin\lb 2\delta\rb}{2\rho}.
    \end{equation}
    Both are hyperbolic of index $+1$ (neither a saddle): $\vb{x}_{+}$ is a sink and
    $\vb{x}_{-}$ is a source. The source radius $R_{-}$ is finite for all parameters; the sink
    radius $R_{+}$ is finite precisely while $B \rho < A$, diverging as $B \rho \to A$. 
    The disc center is a regular point not an equilibrium.
    At $\alpha = 0$, the equilibria degenerate into the neutrally stable centers 
    as in Appendix~B.
    For $\alpha \in (-1, 0)$ the reflection $\Dphi \mapsto -\Dphi$ gives the mirror
    configuration, with the sink and source exchanged.
\end{theorem}

\begin{proof}
Take $\alpha > 0$ (the case $\alpha < 0$ is identical with $U$ and $L$ exchanged). 
Setting $\dot R = \dot\Dphi = 0$ in Eq.~\eqref{eq:aniso_singlet} gives the 
two stated equilibria, and a direct computation gives the Jacobian invariants. Since
$\sin(2\delta) > 0$, we have $\tr J_{+} < 0$ and $\tr J_{-} > 0$; with
$\det J_{\pm} > 0$, $\vb{x}_{+}$ is a sink and $\vb{x}_{-}$ a source,
both non-saddles of index $+1$. \\

\noindent The disk center is not a hidden equilibrium. In $(u, v) = \Phi(R, \Dphi)$ coordinates (Eq.~\eqref{eq:gen_Phi}), a direct limit gives $\vb{F}(*) = \lim_{R\to 0}(\dot u, \dot v) = (\ell A, -\ell B)$ independent of the approach direction, with $\abs{\vb{F}(*)} = \ell\rho \neq 0$; at the center, the friction anisotropy tilts the velocity off the polarity by the body-frame drift angle $\phi_0 = \arctantwo(B, A)$, which vanishes only at $B = 0$. So the only equilibria are $\vb{x}_{\pm}$.

\end{proof}

\subsection{Global behavior}

\begin{theorem}[Global behavior]\label{thm:aniso_global}
    Let $\alpha > 0$ and $\delta \in (0, \pi/2)$. The flow~\eqref{eq:aniso_singlet} has no 
    periodic orbit, and with the source $\vb{x}_{-}$ and (for $B \rho < A$) the sink
    $\vb{x}_{+}$ of Theorem~\ref{thm:aniso}, both of index $+1$, its global behavior 
    splits:
    \begin{enumerate}
    
        \item Bounded ($B\rho < A$). The sink $\vb{x}_{+}$ is finite, and every trajectory in
        $\mathcal{M} \setminus \lcb \vb{x}_{-} \rcb$ converges to it; a single global attractor.

        \item Escape ($B \rho \geq A$). The sink has escaped to infinity, $\vb{x}_{-}$ is the 
        only equilibrium, and every trajectory in $\mathcal{M}\setminus\lcb \vb{x}_{-} \rcb$
        has $R(t) \to \infty$ linearly in time, at a strictly positive asymptotic mean 
        radial speed $m_{\infty}$. The polarity precesses ($B\geq 1$) or locks onto 
        $\Dphi_{\infty} = \arcsin(-B)$ ($B < 1$).
        
    \end{enumerate}
   
\end{theorem}

\begin{proof}
\emph{No periodic orbits.} 
Recall that $v_r = \dot{R} = \ell (A \cos\Dphi - B\sin\Dphi)$ and 
$v_{\theta} = \ell \lb A \sin\Dphi + B\cos\Dphi \rb$, so that 
$\dot{\Dphi} = \sin\Dphi + B - v_{\theta} / R$ and $\partial_{\Dphi} v_r = -v_{\theta}$.
With $g(R) = Re^{-AR/(\ell \rho^2)}$ for which $g' - g/R = -(A/\ell \rho^2) g$, and 
using $\nabla\cdot\vb{F} = \cos\Dphi - v_r / R$, we find
\[
    \nabla\cdot\lb g \vb{F} \rb = g\cos\Dphi + v_r \lb g' - \frac{g}{R} \rb
        = \frac{B}{\ell \rho^2} g(R) v_{\theta}.
\]
Its sign is that of $v_{\theta}$ which is not of one sign. But $g v_{\theta}$ is the
divergence of two vector fields:
\[
    g v_{\theta} = \nabla\cdot\lb \frac{\ell\rho^2}{B} g \vb{F}\rb = \nabla\cdot\lb 0, - gv_r\rb.
\]
Suppose, for contradiction, that the flow has a periodic orbit $\Gamma \subset\mathcal{M}$, 
and let $\Omega$ be the region it bounds. Applying the divergence theorem to each field
over $\Omega$ evaluates $\mathcal{I}$ in two ways:
\[
    0 
    =
    \frac{\ell\rho^2}{B}\oint_{\partial\Omega} g(\vb{F}\cdot\mathbf{n}) \dd s
    =
    \mathcal{I} 
    = 
    \oint_{\partial\Omega} (0, -g v_r) \cdot \mathbf{n} \dd s 
    = 
    \oint_{\Gamma} g v_r^2 \dd t \neq 0.
\]
The left side vanishes because $\mathbf{F}$ is tangent to $\Gamma$. For the right,
the outward normal element is $\mathbf{n} \dd s = (\dd\Dphi, -\dd R)$,
so $(0, -g v_r)\cdot\mathbf{n}\dd s = g v_r \dd R$, then 
$\dd R = v_r \dd t$ on $\Gamma$ and $v_r \neq 0$. This is a contradiction and hence 
no periodic $\Gamma$ can exist, and with two index +1 equilibria and no saddle there are no
homoclinic or heteroclinic connections either, so no $\omega$-limit set is a cycle.

\emph{Bounded regime.} ($B\rho < A$). Here $B \rho < A$ implies $B < 1$. As $R \to \infty$
the orbital term in~\eqref{eq:aniso_singlet} drops and $\dot\Dphi \to \sin\Dphi + B$,
whose stable rest angle has $\sin\Dphi_{\infty} = -B$. There 
$\dot{R} \to \ell( B^2 - A\sqrt{1-B^2})< 0$. Hence beyond some radius the flow turns inward,
a disc $D_{R_{\rm max}}$ is absorbing, and every forward orbit is bounded with non-empty
compact $\omega$-limit set. With no periodic orbit and no saddle connection, 
Poincare-Bendixson leaves $\vb{x}_{+}$ as the only admissible $\omega$-limit set.
Since $\vb{x}_{-}$ is a source every trajectory in $\mathcal{M}\setminus\lcb \vb{x}_{-}\rcb$
converges to $\vb{x}_{+}$.

\emph{Escape regime} ($B\rho \geq A$). Now $R_{+} = \ell\rho^2 / (A - B \rho)$ 
has passed through $+\infty$, leaving the source $\vb{x}_{-}$ as the only equilibrium. 
A source is not any point's $\omega$-limit, and there is no cycle or saddle, so every 
trajectory in $\mathcal{M}\setminus \lcb \vb{x}_{-}\rcb$ is unbounded and $R(t) \to \infty$.

For the rate, as $R \to\infty$ we have $\dot\Dphi \to \sin\Dphi + B$. 
For $B < 1$ the polarity locks at a fixed angle $\Dphi_{\infty}$ ($\sin\Dphi_{\infty}=-B$), 
giving $\dot{R} \to \ell ( B^2 - A \sqrt{1 - B^2}) =: m_{\infty}$. 
For $B \geq 1$ it circulates instead, with $a = 1 - A \ell / R$ and $b = -B \ell / R$, the radius
gained per orbit is
\[
    \Delta R_{\rm frozen}(R) = \oint \frac{\dot{R}}{\dot{\Dphi}} \dd\Delta\phi  = \frac{2\pi \ell B}{a^2 + b^2}
        \lb \frac{B}{\sqrt{B^2 - a^2 - b^2}} - 1 \rb,
\]
which tends as $R \to\infty$ to $m_{\infty}T$, with $T = 2\pi/\sqrt{B^2 - 1}$ and $m_{\infty} = \ell B (B - \sqrt{B^2 - 1})$.
In both cases $m_{\infty} > 0$ when $B \rho > A$, continuous across $B = 1$ ($m_{\infty} = \ell$), 
and $R(t) / t \to m_{\infty}$. In other words, escape is linear in time.
\end{proof}

\section{Chiral confinement}
\label{sec:chiwall}

We now turn on the chiral wall offset $\chi \neq 0$ (a fixed body-frame
rotation of the wall-alignment target) on an isotropic substrate with no anchor ($\alpha = f = \mu = 0$). 
Like the friction offset $\delta$ of Appendix~D, $\chi$ breaks the reflection symmetry of Appendix~B. The offset $\delta$ acts as a state-independent effective torque that \emph{moves} the equilibria; the chiral offset $\chi$ leaves the equilibria where they are and instead \emph{splits their stability}, turning one center into a sink and the other into a source, and so selecting a single global attractor.

Because $\alpha = f = 0$ the velocity is again parallel to the polarity, the velocity-alignment torque drops out, and with $\theta_W = \theta + \pi$ the wall term becomes $\sin\lb\theta_W - \varphi + \chi\rb = \sin\lb\Dphi - \chi\rb$. The fundamental reduced $(R, \Dphi)$ system becomes
\begin{equation}\label{eq:chiwall_reduced}
    \dot R = V_0\cos\Dphi,
    \quad
    \dot{\Dphi} = \frac{\sin\lb\Dphi - \chi\rb}{\tau_W} - \frac{V_0}{R}\sin\Dphi,
\end{equation}
and $\chi = 0$ recovers the symmetric flow in Eq.~\eqref{eq:sym_reduced}. Splitting the chiral torque,
\begin{equation}
    \sin\lb \Dphi - \chi \rb
    = \underbrace{\cos\chi\,\sin\Dphi}_{\text{odd in }\Dphi}
    - \underbrace{\sin\chi\,\cos\Dphi}_{\text{even in }\Dphi},
\end{equation}
shows the effect of $\chi$ on dynamics. The odd piece merely rescales the wall-torque amplitude by $\cos\chi$, but importantly, the even piece breaks $\Dphi \to -\Dphi$. Crucially, this symmetry-breaking piece carries a $\cos\Dphi$ factor, so, unlike the constant
bias $\mu$ of Appendix~C, it \emph{vanishes} at the equilibria
$\Dphi = \pm\pi/2$. It is therefore invisible in the equilibrium positions and appears only in their stability.

\subsection{Equilibria}

Setting $\dot R = 0$ in Eq.~\eqref{eq:chiwall_reduced} forces $\cos\Dphi = 0$, so the interior equilibria stay at $\Dphi = \pm\pi/2$, now at the common radius $R^\star = \frac{\ell}{\cos\chi}$ with $\ell = V_0\tau_W$ and rotation rates $\dot\theta = \pm\cos\chi/\tau_W$. Physically, the chiral confinement interaction puts the cell on a slightly wider orbit ($R^\star > \ell$) that it traces a little more slowly (rate $\propto \cos\chi$) compared to the achiral case, but it does~\emph{not} on its own move the two rotating states or pick a winner, because the symmetry-breaking term is zero there. The selection happens one order up, in the linearization.

\subsection{Global stability}

\begin{theorem}\label{thm:chiwall}
Let $\alpha = f = \mu = 0$ and $\chi \in (0, \pi/2)$. The chiral-wall system in Eq.~\eqref{eq:chiwall_reduced} on $\mathcal{M}$ has exactly two equilibria, both hyperbolic and of index $+1$ (neither a saddle), at the common radius $R^\star = \ell/\cos\chi$:
\begin{equation}
        \vb{x}_{+} = \lb R^\star, -\tfrac{\pi}{2}\rb,
        \quad
        \vb{x}_{-} = \lb R^\star, +\tfrac{\pi}{2}\rb,
\end{equation} with Jacobian invariants
\begin{equation}
        \det J_{\pm} = \lb\frac{\cos\chi}{\tau_W}\rb^2 > 0,
        \quad
        \tr J_{\pm} = \mp \frac{\sin\chi}{\tau_W}.
\end{equation}
Thus, $\vb{x}_{+}$ is a sink and $\vb{x}_{-}$ a source (foci for $\chi < \arctan 2$, nodes for $\arctan 2 < \chi < \tfrac{\pi}{2}$). The flow has no periodic orbit, and every trajectory in $\mathcal{M}\setminus \lcb \vb{x}_{-} \rcb$ converges to $\vb{x}_{+}$: it is the unique global attractor, carrying clockwise rotation $\dot\theta^\star = -\cos\chi/\tau_W$. For $\chi \in \lb -\tfrac{\pi}{2}, 0\rb$, the two equilibria exchange roles. For $\chi = 0$, the equilibria degenerate into the neutrally stable centers of Appendix~B.
\end{theorem}

\begin{proof}
The equilibria and $R^\star$ are found above. Linearizing Eq.~\eqref{eq:chiwall_reduced} at $(R^\star, \pm\pi/2)$, produces
\begin{equation}
    J\big|_{\pm\pi/2} =
        \begin{pmatrix}
            0 & \mp V_0 \\
            \pm V_0/{R^\star}^2 & \pm\sin\chi/\tau_W
        \end{pmatrix}.
\end{equation}
This allows computation of the determinant and trace, respectively: $\det J = \lb V_0/R^\star\rb^2 = \lb\cos\chi/\tau_W\rb^2 > 0,$ and $\tr J = \pm\sin\chi/\tau_W$. We conclude that neither equilibrium is a saddle, and for $\chi > 0$ the sign of the trace makes $\Dphi = +\pi/2$ a source and $\Dphi = -\pi/2$ a sink; the eigenvalues are complex (foci) when $\tr^2 < 4\det$, i.e.\ $\tan^2\chi < 4$. This is precisely where $\chi$ acts: the even term $-\sin\chi\cos\Dphi$, invisible at the equilibria, is what supplies the $\pm\sin\chi/\tau_W$ in the trace.\\

\noindent For the global statement we appeal to Lemma~\ref{lem:toolkit}. As $R \to \infty$ the term $-(V_0/R)\sin\Dphi$ drops out and $\dot\Dphi \to \sin\lb\Dphi - \chi\rb/\tau_W$ drives $\Dphi$ to its stable rest angle $\Dphi = \chi + \pi$, where $\dot R = -V_0\cos\chi < 0$; hence some disk
$D_{R_{\max}}$ is absorbing and all forward orbits are bounded. At the axes, the flow is strictly transverse, $\dot\Dphi = -\sin\chi/\tau_W$ at $\Dphi = 0$ and $+\sin\chi/\tau_W$ at $\Dphi = \pi$, carrying the upper half $U = \lcb 0 < \Dphi < \pi\rcb$ into the lower half $L = \lcb -\pi < \Dphi < 0\rcb$; thus, $L$ is forward-invariant. The multiplier $\eta = 1/\sin\Dphi$ gives $\nabla\cdot\lb\eta\,\vb{F}\rb = \sin\chi/(\tau_W\sin^2\Dphi) > 0$ on each half of the domain, excluding periodic orbits; the only equilibrium $U$ contains is the source $\vb{x}_{-}$, so by Poincar\'e--Bendixson every trajectory except $\vb{x}_{-}$ leaves $U$ and enters $L$ in finite time. On $L$, forward-invariant, simply connected, and holding the single non-saddle equilibrium $\vb{x}_{+}$, Lemma~\ref{lem:toolkit} applies and gives convergence to $\vb{x}_{+}$ equilibrium.
\end{proof}

\section{Center cell-substrate adhesive force}
\label{sec:anchor}


We now turn on the radial force $f(R)$ in Eq.~\ref{eq:gen_master} for an unbiased cell ($\mu = 0$) on an isotropic substrate ($\alpha = 0$) with an achiral wall ($\chi = 0$). Writing $\kappa = k_0/k_W$ for the anchor-to-confinement stiffness strength ratio and setting the disk radius to one, the confinement force of Appendix~A is $f(R) = k_W\lsb 1 - (1 + \kappa) R \rsb$; namely, a restoring force, $f'(R) = -k_W\lb 1 + \kappa\rb < 0$, for $\kappa > -1$.\\

\noindent For $f \neq 0$, the velocity is no longer parallel to the polarity orientation, and therefore the signed velocity-alignment term of Eq.~\eqref{eq:gen_perp} no longer vanishes. Instead, it contributes $-f(R)\sin\Dphi/\lb\gamma_\text{pol}\tau_{\rm VA}\rb$ to $\dot\varphi$. Like the dynamics of the frictional cell-substrate anisotropy in Appendix~D, this surviving torque makes the flow
dissipative and destroys the conserved Hamiltonian quantity. Unlike the
chiral confinement interaction in Appendix~E, however, it does so \emph{symmetrically}: the two rotating states are not split in stability, and each instead becomes an attracting spiral, leaving the system perfectly balanced between them.\\

\noindent Specializing the fundamental reduced system in Eq.~\eqref{eq:gen_master} ($\alpha = \delta = \chi = \mu = 0$) gives the reduced flow on $\mathcal{M}$,
\begin{equation}\label{eq:anchor_reduced}
    \dot R = \underbrace{V_0\cos\Dphi + \frac{f(R)}{\xi}}_{P(R, \Dphi)},
    \quad
    \dot{\Dphi} = \underbrace{\sin\Dphi\, G(R)}_{Q(R, \Dphi)},
\end{equation}
with the $R$-dependent rate
\begin{equation}\label{eq:anchor_G}
\begin{aligned}
    G(R) &= \frac{1}{\tau_W} - \frac{V_0}{R} - \frac{f(R)}{\gamma_\text{pol}\tau_{\rm VA}},\\
    G'(R) &= \frac{V_0}{R^2} + \frac{k_W\lb 1 + \kappa\rb}{\gamma_\text{pol}\tau_{\rm VA}} > 0.
\end{aligned}
\end{equation}
The decisive feature is that the whole $\Dphi$-dependence of $\dot{\Dphi}$ is carried by the prefactor $\sin\Dphi$, with $G$ a function of $R$ alone. This feature is what keeps the two directions of rotation exactly even and makes the global structure rigid.

\begin{theorem}[Global bistability]\label{thm:anchor}
Let $\alpha = \chi = \mu = 0$ and $f \neq 0$ with $\kappa > -1$ (so $f' < 0$). The reduced flow in Eq.~\eqref{eq:anchor_reduced} on $\mathcal{M}$ has the following structure:
\begin{enumerate}
    \item \textbf{Invariant axes.} The lines $\Dphi = 0$ and $\Dphi = \pi$ are invariant and split $\mathcal{M}$ into two strips
    $\mathcal{S}^{+} = \lcb 0 < \Dphi < \pi \rcb$ and $\mathcal{S}^{-} = \lcb -\pi < \Dphi < 0 \rcb$, exchanged by the reflection $\Dphi \to -\Dphi$;

    \item \textbf{One sink per strip.} Each strip contains exactly one interior equilibrium, at the common radius $R^\star$ fixed by $G(R^\star) = 0$ and $\Dphi^{\star}_{\pm} = \pm\arccos\lb -f(R^\star)/\gamma_\text{pol} \rb$. It is an asymptotically stable spiral/node,
    \begin{equation}
        \tr J = \frac{f'(R^\star)}{\xi} < 0,
        \quad
        \det J = V_0 \sin^2\Dphi^{\star}\, G'(R^{\star}) > 0,
    \end{equation}
    and the two sinks are mirror images, with $\dot\theta = V_0\sin\Dphi^{\star}/R^{\star}$ of equal magnitude and opposite
    sign;

    \item \textbf{Axis saddles.} On the invariant lines lie hyperbolic saddles at $(R_{\rm out}, 0)$ with $f(R_{\rm out}) = -\gamma_\text{pol}$, and (when $k_W > \gamma_\text{pol}$) at $(R_{\rm in}, \pi)$ with $f(R_{\rm in}) = \gamma_\text{pol}$;

    \item \textbf{No periodic orbits.} The Dulac multiplier $\eta = 1/\sin\Dphi$ gives $\nabla\cdot\lb\eta\,\vb{F}\rb = f'(R)/\lb \xi \sin\Dphi \rb$, of one strict sign on each strip.
\end{enumerate}
Consequently the system is globally bistable: every trajectory off the axes converges to one of the two mirror-image sinks, the basins separated by the stable manifolds of the axis saddles. An unbiased ensemble therefore splits 50/50 between CW and CCW directions.
\end{theorem}

\begin{proof}
On $\Dphi \in \lcb 0, \pi \rcb$ we have $\sin\Dphi = 0$, so
$Q = \sin\Dphi\, G(R) = 0$ and the axes are invariant; the reflection
$\Dphi \to -\Dphi$ sends $(P, Q) \to (P, -Q)$, exchanging the strips.\\

\noindent For the equilibria, $\dot R = 0$ gives $\cos\Dphi^\star = -f(R^\star)/\gamma_\text{pol}$ (using $\xi V_0 = \gamma_\text{pol}$).
Off the axes $\dot\Dphi = 0$ forces $G(R^\star) = 0$, which has a unique root because $G' > 0$; the two angles $\Dphi^\star_\pm = \pm\arccos\lb -f(R^\star)/\gamma_\text{pol}\rb$ place one
interior equilibrium in each strip. On the axes instead ($\sin\Dphi = 0$), $\dot R = 0$ needs $f(R) = \mp\gamma_\text{pol}$, giving the axis equilibria $R_{\rm out}$ ($\Dphi = 0$) and $R_{\rm in}$ ($\Dphi = \pi$, present when $k_W > \gamma_\text{pol}$).\\

\noindent The Jacobian in Eq.~\eqref{eq:anchor_reduced} simplifies at every equilibrium because its angular entries $Q_R = \sin\Dphi\, G'(R)$ and $Q_\Dphi = \cos\Dphi\, G(R)$ vanish in turn. At an interior equilibrium, $G(R^\star) = 0$ cancels $Q_\Dphi$, so
$\tr J = P_R = f'(R^\star)/\xi < 0$ and
$\det J = -P_\Dphi Q_R = V_0\sin^2\Dphi^\star\, G'(R^\star) > 0$: an asymptotically stable spiral or node. At an axis equilibrium. $\sin\Dphi = 0$ cancels $Q_R$, so $\det J = P_R\, Q_\Dphi = \lb f'(R)/\xi\rb\cos\Dphi\, G(R)$, a saddle ($\det J < 0$) at both crossings: they straddle the interior radius, since
$f(R_{\rm out}) = -\gamma_\text{pol} < f(R^\star) < \gamma_\text{pol} = f(R_{\rm in})$ with $f$ decreasing gives $R_{\rm in} < R^\star < R_{\rm out}$, and $G$ increasing gives $G(R_{\rm in}) < 0 < G(R_{\rm out})$.\\

\noindent Finally, with $\eta = 1/\sin\Dphi$, $\nabla\cdot\lb\eta\,\vb{F}\rb = P_R/\sin\Dphi = f'(R)/\lb\xi\sin\Dphi\rb$, which is $< 0$ on $\mathcal{S}^{+}$ and $> 0$ on $\mathcal{S}^{-}$. As $R \to \infty$, $\dot R \approx f(R)/\xi \to -\infty$, so some disk $D_{R_{\max}}$ is absorbing and all forward orbits are bounded. Each open strip is forward-invariant (the axes are
invariant), simply connected, and contains the single non-saddle equilibrium of its sink, so Lemma~\ref{lem:toolkit} applies on each strip: every trajectory in $\mathcal{S}^{\pm}$ converges to its sink, and the axis saddles with their stable manifolds partition the basins.
\end{proof}

\subsection*{Combining a symmetry-breaking mechanism with the adherent center anchor}

On its own, the anchor is globally bistable (Theorem~\ref{thm:anchor}): a CW and a CCW sink with equal basins, so an unbiased ensemble splits 50/50. Switching on any one of the remaining couplings (the intrinsic bias $\mu$, the chiral wall $\chi$, or the anisotropic friction $\alpha,\delta$) breaks that reflection and makes the two basins of attraction unequal: one sink is favored, its basin grows, and the ensemble handedness tilts continuously with the coupling strength. Only at the extreme does the disfavored basin close and the population become one-handed. Each coupling is thus a tunable mechanism for the CW/CCW distribution. Details are provided below.

Both sinks stay stable as the coupling turns on. At an anchored sink
$G(R^\star) = 0$, so the angular part of the Jacobian vanishes and only the radial stiffness survives in the trace; a coupling reinstates an angular contribution, and linearizing the flow in Eq.~\eqref{eq:anchor_reduced} at an interior sink (where
$\cos\Dphi^\star = -f(R^\star)/\gamma_\text{pol}$) gives
\begin{equation}\label{eq:anchor_trace}
    \tr J =
    \begin{cases}
        \dfrac{f'(R^\star)}{\xi}, & \text{anchor alone},\\[8pt]
        \dfrac{f'(R^\star)}{\xi} - \mu\cot\Dphi^\star, & \text{with intrinsic bias } \mu,\\[8pt]
        \dfrac{f'(R^\star)}{\xi} + \dfrac{\sin\chi}{\tau_W \sin\Dphi^\star}, & \text{with chiral wall } \chi.
    \end{cases}
\end{equation}
The radial term $f'(R^\star)/\xi < 0$ always pulls the trace in the negative direction; whether a coupling can overturn it decides the disfavored sink's outcome. The basin boundary is the stable manifold of the axis saddle, so tilting the basins amounts to sweeping that separatrix toward the disfavored sink, which can then disappear in one of two ways: it collides with the axis saddle and annihilates in a saddle-node bifurcation, or, if the coupling drives its trace through zero, it loses stability in a Hopf bifurcation. Which one of these routes emerges depends on the symmetry breaking mechanism (and is discussed below).
\begin{itemize}[itemsep=0pt]
\item\emph{\underline{Intrinsic bias}} ($\mu$) enters $\dot\Dphi$ additively and slides the equilibria (now $G(R^\star) = -\mu/\sin\Dphi^\star \neq 0$) without destabilizing either sink, so the disfavored state is lost by the saddle-node bifurcation route. The continuation in $(\mu, k_0)$ is a bistable split bounded by saddle-node curves; strong enough $\mu$ exits it into a unique globally stable CW or CCW state (Fig.~\ref{fig:anchor}, top row).

\item\emph{\underline{Anisotropic friction}} $(\alpha,\delta)$ tilts the basins the same way as the bias, but its reach is limited. The offset $\delta$ is a state-independent effective torque that slides the equilibria off $\pm\pi/2$ without destabilizing a sink, and within the physical range $\alpha > -1$ it is too weak to drive the disfavored sink into the saddle. Both states therefore survive: the $(\alpha,k_0)$ plane stays bistable almost everywhere, skewing the ensemble rather than making it one-handed (Fig.~\ref{fig:anchor}, bottom row), with non-generic behavior confined to the edge $\alpha \to -1$ where the mobility degenerates.

\item\emph{\underline{Chiral confinement}} ($\chi$) is the exception. Expanding $\sin\lb\Dphi - \chi\rb$ produces a parity-even $\cos\Dphi$ term that the bias and the mobility do not supply, and it adds the $\sin\chi/(\tau_W\sin\Dphi^\star)$ to the trace in Eq.~\eqref{eq:anchor_trace}. A strong enough $\chi$ then drives the disfavored sink's trace through zero, closing its basin by the Hopf bifurcation route instead of the saddle-node (Fig.~\ref{fig:anchor}, middle row).
\end{itemize}